\documentclass[a4paper,12pt]{article}
\usepackage{blindtext}
\usepackage[utf8]{inputenc}
\usepackage[T1]{fontenc}
\usepackage{graphicx}
\usepackage{color}
\usepackage{amsfonts}
\usepackage{amsmath}
\usepackage{amssymb}
\usepackage{amsthm}
\usepackage{proof}
\usepackage[hidelinks]{hyperref}

\newtheorem{defi}{Definition}

\newtheorem{prop}{Proposition}

\newtheorem{teo}{Theorem}
\providecommand{\keywords}[1]{\textit{Keywords:  } #1}

\newcommand{\grmor}{\;|\;}
\newcommand{\skipk}{\textsf{skip}}
\newcommand{\assigk}[2]{\textsf{#1:=#2}}
\newcommand{\loadk}[2]{\textsf{#1:=[#2]}}
\newcommand{\muttk}[2]{\textsf{[#1]:=#2}}
\newcommand{\consk}[2]{\textsf{#1:=cons(#2)}}
\newcommand{\dispk}[1]{\textsf{disp(#1)}}
\newcommand{\seqk}[2]{#1\textsf{\:;\:}#2}
\newcommand{\park}[2]{#1\textsf{\:}\|\textsf{\:}#2}
\newcommand{\ifk}[3]{\textsf{if\;}#1\textsf{\;then\;}#2\textsf{\;else\;}#3}
\newcommand{\whik}[2]{\textsf{while\;}#1\textsf{\;do\;}#2}
\newcommand{\resk}[2]{\textsf{resource\;}#1\textsf{\;in\;}#2}
\newcommand{\withk}[3]{\textsf{with\;}#1\textsf{\;when\;}#2\textsf{\;do\;}#3}
\newcommand{\withkk}[2]{\textsf{with\;}#1\textsf{\;do\;}#2}
\newcommand{\withink}[2]{\textsf{within\;}#1\textsf{\;do\;}#2}
\newcommand{\abok}{\textsf{abort}}
\newcommand{\empk}{\texttt{emp}}

\title{Revisiting Concurrent Separation Logic}

\author{Pedro Soares\thanks{Universidade do Porto, PT}, António Ravara\thanks{CITI \& DI-FCT, Universidade Nova de Lisboa, PT} and Simão Melo de Sousa\thanks{LISP \& LIACC\& DI-FE, Universidade da Beira Interior, PT
\newline
This work was partially funded by Funda\c{c}\~{a}o para a Ci\^{e}ncia e Tecnologia through AVIACC project, grant PTDC/EIA-CCO/117590, and CITI/FCT/UNL, grant Pest-OE/EEI/UI0527/2014.
}}

\date{\vspace{-13mm}}

\begin{document}

\maketitle

\begin{abstract}
  We present a new soundness proof of Concurrent Separation Logic
  (CSL) based on a structural operational semantics (SOS). We build on
  two previous proofs and develop new auxiliary notions to achieve the
  goal. One uses a denotational semantics (based on traces). The other
  is based on SOS, but was obtained only for a fragment of the logic
  --- the Disjoint CSL --- which disallows modifying shared variables
  between concurrent threads.  In this work, we lift such a
  restriction, proving the soundness of full CSL with respect to a
  SOS. Thus contributing to the development of tools able of ensuring
  the correctness of realistic concurrent programs. Moreover, given
  that we used SOS, such tools can be well-integrated in programming
  environments and even incorporated in compilers.
\end{abstract}

\vspace{1em}
\hspace{-1.8em}

\section{Introduction}

The aim of this work is to present a new soundness proof for Concurrent Separation Logic \cite{hearn}, with respect to a structural operational semantics~\cite{plot}. This work adapts and extends the results presented by Brookes~\cite{bro2} and by Vafeiadis~\cite{vaf}.


The axiomatic verification of programs goes back to Hoare Logic \cite{hoar}. This seminal work introduces two key ideas, i) the specification of programs by means of what is known by a Hoare triple: $\{P\}C\{Q\}$, where $P$ and $Q$ are first order formulae, called the precondition and postcondition respectively, and $C$ is an imperative program; ii) a deductive proof system to ensure the partial correctness of programs. A program is partially correct, if every execution of $C$ from a state respecting the precondition does not abort and when it terminates the postcondition holds for its final state.  The state for this logic is formed only by the store, i.e. a partial function that records the value of each variable. Hoare's work gave rise to numerous deductive systems, for instance the Owicki-Gries method (\cite{owi, owigrie}) and Separation Logic (\cite{rey2, rey}).


The Owicki-Gries method is one of the first attempts to give a resource sensitive proof system for concurrent programs. To do this, Owicki and Gries augmented the programming language with i) parallel composition, $\park{C}{C}$; ii) local resources, $\resk{r}{C}$; and iii) a critical region, $\withk{r}{B}{C}$, where $r$ denotes a resource. Each resource has a mutual exclusion lock, an assertion, called invariant, and a set of variables, called protected variables.

The execution of parallel composition non-deterministically chooses one of the commands to execute first. As usual, the parallel execution is assumed to be weakly fair, i.e. if a command is continually available to be executed, then this command will be eventually selected.
The resource command declares a local variable $r$ to be used in $C$. 
The critical region command waits for the availability of the resource $r$,  and when $B$  holds,  it acquires $r$ and starts the execution of $C$; the resource $r$ is released upon the execution of $C$ terminates.


The programs derivable by the Owicki-Gries method have to preserve the resource invariants when the resource is free, and respect the protection of variables by resources, i.e. a program needs to acquire all resources protecting a variable, before the program can change that variable.
The parallel rule proposed by Owicki \cite{owi} requires that every variable occurring in the derivation proof of one command cannot be changed by another command, except for variables protected by a resource such that the variables only appear inside the critical region's proof. Thus, the Owicki-Gries method is not compositional.




Separation Logic (SL) supports reasoning about imperative programs with shared mutable data and consequently about dynamical data structures, such as lists and trees.
In order to do this, the assertion and program languages used by Hoare had to be augmented.
The assertions are extended with the constructs $\empk$, the empty memory; $e\mapsto e'$, a single memory cell $e$ with the value $e'$; and $P\ast Q$, two disjoint memory's parts such that one satisfies $P$ and the other satisfies $Q$. In this settings, the memory is usually represented by the heap --- a partial function from the set of locations to the set of values. The store and the heap together define the state of a program.

The programing language is augmented with commands for memory manipulation. Naturally, the proof system is also extended with a rule for each new command and with a frame rule, used to enlarge the portion of memory considered in the condition of a specification. This rule is crucial to achieve local reasoning: program specifications only need to consider the relevant memory for their execution.
Therefore, this local reasoning mechanism can be used to establish the partial correctness of disjoint concurrent programs, i.e. concurrent program which does not change shared variables.

In order to prove the soundness of the frame rule, and thus of local reasoning, it is sufficient to ensure the validity of two key properties: safety monotonicity and the frame property. Safety monotonicity states that if an execution does not abort for a given memory portion, then the execution does not abort for any memory portion that contains the initial one. The frame property says that if a command does not abort for a given memory portion, then every execution on a larger memory corresponds to an execution on the initial memory.

Recently provers based on separation logic were adopted in real industrial projects, Facebook's infer being the most prominent of such tools \cite{CDDGHLOP15}.

Since the introduction of SL, different authors adapted it to the verification of concurrent programs. Vafeiadis and Parkinson introduced RGSep, combining SL with Rely/Guarantee reasoning \cite{vaf2}. Reddy and Reynolds introduced a syntactic control of interference in SL \cite{red}, borrowing ideas from works on fractional permissions \cite{boy}. O'Hearn  proposed Concurrent Separation Logic (CSL), combining SL with the Owicki-Gries method \cite{hearn}.
Brookes formalized CSL, extending the traditional Hoare triples with a resource context $\Gamma$ and a rely-set $A$, what leads to specifications of the form $\Gamma\models_{A}\{P\}C\{Q\}$.
A resource context records the invariant and the protected variables of each resource.
A rely-set consists of all variables relevant for its derivation tree. This set ensures that CSL is a compositional proof method, proved sound with respect to a denotational semantics based on traces, where a program state is represented by a store, a heap and sets of resources, expressing resource ownership \cite{bro2}. Actually, the rely-set was introduced after Wehrman and Berdine discovered a counter-example to the initial version of CSL \cite{bro1}, and it is analogous to the set of variables used by Owicki and Gries to check non-interference in their parallel rule.

Alternatively, Vafeiadis proposed a structural operational semantics (SOS) for concurrent programs synchronizing via resources, and proved the soundness of a part of CSL, the Disjoint CSL (DCSL) \cite{vaf}. DCSL and CSL have different side conditions for the parallel rule. Concurrent threads, in DCSL, must not modify all variables that appear in other Hoare triples, however concurrent threads, in CSL, can not modify variables that belongs to other rely-sets.

Our aim is to remove the disjointness condition and obtain a soundness proof using a SOS for the full CSL (Section \ref{soundness}). The goal is relevant because CSL has been adopted as the basis for most modern program logics, and it is a step in the development of more expressive provers well integrated in software development environments and compilers.  Not only does it allows proving correct concurrent programs manipulating shared resources, but also provides techniques to equip compilers with mechanisms of detecting data-races.
Concretely, the contributions of this work are the following:
\begin{enumerate}
	\item A novel notion of environment transition that simulates actions made by other threads. We define it taking into account the rely-set, available resources and their invariants (Section \ref{envtrasec}). This relation is crucial to study the 
	soundness of the parallel rule;%
        \footnote{Vafeadis used a completely different notion of environment transition (in RGSep \cite{vaf2}).}
	\item The resource configuration that expresses ownership. It is defined by three sets:  owned resources, locked resources, and available resources (Section \ref{progtrans}). A program state is formed by a store, a heap and a resource configuration. Brookes also used sets of resources to represent resources ownership~\cite{bro2};
\item Illustrative examples that we prove correct in CSL, showing the proof system's expressiveness (Section \ref{motiexa}).
\end{enumerate}

This paper is an extended version of \cite{pdp-4pad15}. 
We present herein more examples and sketches of the proofs of the results reported in the short paper (which does not present proofs).
Further examples and proofs in full detail can be found in a technical report~\cite{RR-DCC-11-2014}.

The paper is organized as follows: first, we review the syntax of concurrent resource-oriented programs with shared mutable data (Section \ref{prosubsec}) and Concurrent Separation Logic proof system (Section \ref{rules}), following the work of Brookes \cite{bro2}. Next, we present a structural operational semantics for the previous programs (Section \ref{progtrans}), along the lines of the work of Vafeiadis \cite{vaf}. We state important results over this operational semantics for the soundness proof, including safety monotonicity and frame property (Section \ref{smfp}). Afterwards, we introduce the environment transition (Section \ref{envtrasec}). Finally, we prove the soundness of Concurrent Separation Logic with respect to the operational semantics we defined (Section \ref{soundness}).
\section{Concurrent Separation Logic}\label{ch2}

We revisit Concurrent Separation Logic (CSL), as presented by Brookes \cite{bro2}. First, we define the assertion language, then the syntax of commands for concurrent programs, and finally the inference rules for CSL.

\subsection{Assertion Language}

Consider a set \textbf{Var} of \emph{variables}, ranged over by $\textsf{x},\textsf{y},\ldots$, and a set \textbf{Val} of \emph{values}, that includes the integers and the value $null$. These meta-variables may be indexed or primed.

\begin{figure}[h]
\begin{eqnarray*}
  e &\equiv& x\grmor n\grmor e_1+e_2\grmor e_1-e_2\grmor e_1\times e_2\\
  B &\equiv& \texttt{true} \grmor \texttt{false} \grmor e_1=e_2\grmor e_1<e_2\grmor  B_1\wedge B_2 \grmor \neg B\\
  P &\equiv& B \grmor \neg P \grmor P_1 \wedge P_2 \grmor \forall x P \grmor \empk \grmor e \mapsto e'_1,e'_2,\dots,e'_n \grmor P_1 \ast P_2 
\end{eqnarray*}
\caption{Syntax of the Assertion Language}
\label{fig:csl}
\end{figure}
The grammar in Figure~\ref{fig:csl} defines the syntax of the assertion language, where $e \mapsto e'_1,e'_2,\dots,e'_n$ denotes $e \mapsto e'_1\ast (e+1) \mapsto e'_2\ast\dots \ast (e+n-1)\mapsto e'_n$.
We assume the usual definitions of \emph{free variables} of an assertion (FV).

We use the definition of SL for the \emph{validity} of an assertion with respect to the pair $(s,h)$, where $s$ and $h$ are denoted by \emph{storage} and \emph{heap}, respectively, and given by the functions:
$$s:\textbf{Var}\rightarrow \textbf{Val}, \quad h:\textbf{Loc} \rightharpoonup \textbf{Val},$$
where $\textbf{Loc}\subset\mathbb{N}$ is the set of current locations, for details see e.g. \cite[Section 2]{rey2}. The set of those pairs is denoted by $\mathcal{S}$.

For an assertion $P$, we write $s,h \models P$ if the assertion is valid for $(s,h)\in \mathcal{S}$, and we write $\models P$ if $s,h\models P$ for every $(s,h)\in\mathcal{S}$. We state a popular result about the validity of assertion, see e.g. \cite[Proposition 4.2]{vaf}.

\begin{prop}\label{astneval}
Let $P$ be an assertion, $(s,h),(s',h)\in\mathcal{S}$.
If $s(x)=s'(x)$, for every $x\in FV(P)$, then $$s, h \models P \ \  \textrm{iff}\ \ s', h \models P.$$
\end{prop}

For a given heap and storage, the precise assertions uniquely determine the subheap that verifies it. The heap $h'$ is a \emph{subheap} of $h$, if the domain of $h'$ is contained in the domain of $h$ and the evaluation of $h'$ coincides with the evaluation of $h$ .
\begin{defi}\label{defastpre}
 We say that an assertion $P$ is \emph{precise} if for every $(s,h)\in \mathcal{S}$ there is at most one subheap $h'$ of $h$ such that $s,h'\models P$.
\end{defi}

Let $\textbf{Res}$ be the set of resources names, which is disjoint from \textbf{Var}. The \emph{resource context} $\Gamma$ is used to represent a shared state. The resource context $\Gamma$ has the form 
\begin{equation}\label{resform}
r_1(X_1):R_1,r_2(X_2):R_2,\ldots, r_n(X_n):R_n,
\end{equation}
where $r_i\in \textbf{Res}$ are distinct, $R_i$ are precise assertions and $X_i\subseteq \textbf{Var}$ such that $FV(R_i)\subseteq X_i$, for each $i=1,2,\ldots, n$. The assertion $R_i$ represents a resource invariant. Since CSL has the conjunction rule, the assertions $R_i$ must be precise, as exemplified by Reynolds \cite[Section 11]{hearn}.

Let $Res(\Gamma)$ denote the set of \emph{resources names} appearing in $\Gamma$, ranged over by $r_i$.
Furthermore, let $PV(\Gamma)$ denote the set of all \emph{protected variables} by resources in $\Gamma$, and $PV(r_i):=X_i$ denote the set of protected variables by $r_i$.

\subsection{Programming Language}\label{prosubsec}

The language includes the \emph{basic commands} to manipulate storage and heap:
$$c\equiv\assigk{x}{e} \grmor \loadk{x}{e} \grmor \muttk{e}{e'} \grmor \consk{x}{$\overline{\textsf{e}}$}\grmor \dispk{e}.$$
The basic commands use the notation of SL. The bracket parenthesis denotes an access to a heap location. For a vector $\overline{\textsf{e}}=(\textsf{e}_1,\dots,\textsf{e}_n)$, $\consk{x}{$\overline{\textsf{e}}$}$ allocates $n$ sequential locations with values $\textsf{e}_1,\dots,\textsf{e}_n$. And  $\dispk{e}$ frees a location.

The following grammar defines the syntax of the \emph{programming language}, $\mathcal{C}$.
\begin{eqnarray*}
  C &\equiv& \skipk \grmor c \grmor  \seqk{C_{1}}{C_{2}} \grmor \ifk{B}{C_1}{C_2} \grmor \whik{B}{C} \grmor \\
 & &  \resk{r}{C} \grmor \withk{r}{B}{C} \grmor \park{C_1}{C_2}
\end{eqnarray*}

The set of \emph{modified variables} by a program $C$, $mod(C)$, consists of all variables $\textsf{x}$ such that the program $C$ has one of the following commands: $\assigk{x}{e}$, $\loadk{x}{e}$ or $\consk{x}{e}$.

 The set of \emph{resources occurring} in a command $C$ is denoted by $Res(C)$ and it consists of all resources names $r$ such that $C$ has one of the following commands: $\withk{r}{B}{C}$, $\resk{r}{C}$.

The substitution on $C$ of a resource name $r$ for a resource name $r'$ not occurring in $C$ is denoted by $C[r'/r]$. 

The set of auxiliary variables have been useful to deduce more specific post conditions for a program's specification, see e.g. \cite{owigrie}. Those variables do not impact the flow of the program. Next, we give the definition of auxiliary variables for a command.

\begin{defi}
Let $C\in \mathcal{C}$. We say that $X$ is a set of \emph{auxiliary variables} for $C$ if every occurrence of $x\in X$ in $C$ is inside an assignment to a variable in $X$. 
\end{defi}

After we have used the auxiliary variables to deduce a specification, we want to remove them from the program. We replace every assignment to an auxiliary variable by the command $\skipk$. This replacement is denoted by $C\setminus X$.

\subsection{Inference rules}\label{rules}

In this section, we present the most relevant inference rules for CSL as stated by Brookes \cite{bro2}.
First, we define what is a well-formed specification in CSL.
\begin{defi}
Let $\Gamma$ be a resource context, $A\subset \textbf{Var}$, $P,Q$ assertions and $C\in \mathcal{C}$. The \emph{specification of a program} has the form 
$$\Gamma \vdash_{A} \{P\} C\{Q\}.$$

Moreover we say that the specification of the program is \emph{well-formed}, if $FV(P,Q)\subseteq A$ and $FV(C)\subseteq A \cup PV(\Gamma)$.
\end{defi}

\begin{figure*}[h]
$$
\infer[\mathrm{(SKP)}]{\Gamma \vdash_{A} \{P\} \skipk \{P\}}{}
\quad
\infer[\mathrm{(SEQ)}]{\Gamma \vdash_{A_1\cup A_2} \{P_1\} \seqk{C_1}{C_2} \{P_3\}}{\Gamma \vdash_{A_1} \{P_1\} C_1 \{P_2\} & \Gamma \vdash_{A_2} \{P_2\} C_2 \{P_3\}}
$$
\vspace{1mm}
\resizebox{\textwidth}{\height}{$$
\infer[\mathrm{(BC)}]{\Gamma \vdash_{A} \{P\} c \{Q\}}{mod(c)\notin PV(\Gamma) & \vdash^{SL} \{P\} c \{Q\}}
\quad
\infer[\mathrm{(FRA)}]{\Gamma \vdash_{A\cup FV(R)} \{P\ast R\} C \{Q\ast R\}}{\Gamma \vdash_{A} \{P\} C \{Q\} & mod(C)\cap FV(R)=\emptyset}
$$}

\resizebox{\textwidth}{\height}{$$\infer[\mathrm{(LP)}]{\Gamma \vdash_{A} \{P\} \whik{B}{C} \{P\wedge \neg B\}}{\Gamma \vdash_{A} \{P\wedge B\} C \{P\}}
\quad
\infer[\mathrm{(CONJ)}]{\Gamma \vdash_{A_1\cup A_2} \{P_1\wedge P_2\} C \{Q_1\wedge Q_2\}}{\Gamma \vdash_{A_1} \{P_1\} C \{Q_1\} & \Gamma \vdash_{A_2} \{P_2\} C \{Q_2\}}
$$}
$$
\infer[\mathrm{(IF)}]{\Gamma \vdash_{A_1\cup A_2} \{P\} \ifk{B}{C_1}{C_2} \{Q\}}{\Gamma \vdash_{A_1} \{P\wedge B\} C_1 \{Q\} & \Gamma \vdash_{A_2} \{P\wedge \neg B\} C_2 \{Q\}}$$
$$\infer[\mathrm{(CONS)}]{\Gamma \vdash_{A'} \{P'\} C \{Q'\}}{\Gamma \vdash_{A} \{P\} C \{Q\} & \models P'\Rightarrow P & \models Q\Rightarrow Q' & A\subseteq A'}$$
\resizebox{\textwidth}{\height}{$$\infer[\mathrm{(AUX)}]{\Gamma \vdash_{A} \{P\} C\setminus X \{Q\}}{\Gamma \vdash_{A\cup X} \{P\} C \{Q\} & X\cap FV(P,Q)=\emptyset & X\cap PV(\Gamma)=\emptyset & X\textrm{ is auxiliary for } C}$$}
$$
\infer[\mathrm{(REN)}]{\Gamma \vdash_{A} \{P\} C \{Q\}}{\Gamma[r'/r] \vdash_{A} \{P\} C[r'/r] \{Q\} & r' \notin Res(C) & r'\notin Res(\Gamma) }
$$
$$
\infer[\mathrm{(PAR)}]{\Gamma \vdash_{A_1\cup A_2} \{P_1\ast P_2\} \park{C_1}{C_2} \{Q_1\ast Q_2\}}{\Gamma \vdash_{A_1} \{P_1\} C_1 \{Q_1\} &  \Gamma \vdash_{A_2} \{P_2\} C_2 \{Q_2\} &  (\ref{parsc})}
$$
\resizebox{\textwidth}{\height}{$$\infer[\mathrm{(CR)}]{\Gamma, r(X):R \vdash_{A} \{P\} \withk{r}{B}{C} \{Q\}}{\Gamma \vdash_{A\cup X} \{(P \wedge B)\ast R\} C \{Q\ast R\}}
\quad
\infer[\mathrm{(RES)}]{\Gamma \vdash_{A\cup X} \{P\ast R\} \resk{r}{C} \{Q\ast R\}}{\Gamma, r(X):R \vdash_{A} \{P\} C \{Q\}}
$$}
\caption{Rules of the Inference System}
\label{fig:infrules}
\end{figure*}
In Figure~\ref{fig:infrules}, we present some inference rules of CSL. The inference rules are only applied for well-formed specifications. The specifications derivable by SL are denoted by $\vdash^{SL}$.


The rule for basic commands are inherited from SL by adding the rely-set (containing all relevant variables for the derivation) and imposing that protected variables are not modified.
The sequential and frame rules are very similar to the respective rules of SL, but the rely-set needs to take into account the rely-sets of both programs or the variables of the framed assertion.

In the critical region rule, if the command inside the critical region preserves the invariant, when $B$ is initially respected, then the resource context can be expanded by $r$. Note that the rely-set does not need to include all protected variables, however the well-formedness of the specification must be preserved.
In the local resource rule, we are able to take out a resource from the assumption's resource context to the conclusion's local condition.
The parallel rule $\mathrm{(PAR)}$ has the side condition below that restricts the interference between programs. 
\begin{equation}\label{parsc}
mod(C_1)\cap A_2=mod(C_2)\cap A_1=\emptyset.\tag{*}
\end{equation}

In order to obtain the inference rules of DCSL we erase the rely-set from the CSL inference rules and change the side condition in the parallel rule $\mathrm{(PAR)}$ to the following condition:
$$mod(C_1)\cap FV(P_2,C_2,Q_2)=mod(C_2)\cap FV(P_1,C_1,Q_1)=\emptyset.$$

Note that every valid specification in DCSL is also valid in CSL, considering that $A_i=FV(P_i,C_i,Q_i)$, for $i=1,2$. 

\section{Motivating examples}\label{motiexa}

\subsection{Semaphore}\label{semexa}

In this example, we present a simple binary semaphore for two threads. Similar examples were studied in \cite[Section\; 4]{hearn}. 
We use the resource invariant to infer the properties of mutual exclusion, absence of deadlocks and starvation. 

This example is a solution to the critical region problem in \cite[Section 3]{ben}.
In contrast to the usual solutions, we
obtain a simpler solution for the critical region problem, due to the
command $\withk{r}{B}{C}$.

We have the following specifications for the thread $p$:

$$se(\textsf{p},\textsf{q}): S \vdash_{\emptyset} \{\empk \}P(p)\{\empk\},$$
$$se(\textsf{p},\textsf{q}): S \vdash_{\emptyset} \{\empk \}V(p)\{\empk\},$$
where 
\begin{itemize}
	\item $S\equiv [(\textsf{p}=0\wedge  \textsf{q}=0)\vee (\textsf{p}=1\wedge  \textsf{q}=0)\vee(\textsf{p}=0\wedge  \textsf{q}=1)] \wedge \empk$,
	\item $P(p) \equiv\withk{se}{\textsf{q}=0}{\assigk{p}{1}}$,
	\item $V(p)\equiv \withkk{se}{\assigk{p}{0}}$.
\end{itemize}

Consider the next well-formed specification of programs.

\begin{eqnarray*}
&\vdash_{\{\textsf{p},\textsf{q}\}}&\{(\empk \wedge \textsf{q}=0)\ast S\}\\&&
\{ \textsf{q}=0 \wedge \empk\}\\&&
\assigk{p}{1}\\&&
\{ \textsf{p}=1 \wedge \textsf{q}=0 \wedge \empk\}\\&&
\{\empk \ast S\}
\end{eqnarray*}

and

\begin{eqnarray*}
&\vdash_{\{\textsf{p},\textsf{q}\}}&\{\empk \ast S\}\\&&
\{  [(0=0\wedge  \textsf{q}=0)\vee (0=0\wedge  \textsf{q}=1)] \wedge \empk\}\\&&
\assigk{p}{0}\\&&
\{ [(\textsf{p}=0\wedge  \textsf{q}=0)\vee(\textsf{p}=0\wedge  \textsf{q}=1)] \wedge \empk\}\\&&
\{\empk \ast S\}
\end{eqnarray*}

Applying the $(CR)$ rule we obtain the desired
specifications. Considering the analogous programs and derivations for
the thread $q$ we obtain:
$$se(\textsf{p},\textsf{q}): S \vdash_{\emptyset} \{\empk \}P(q)\{\empk\},$$
	$$se(\textsf{p},\textsf{q}): S \vdash_{\emptyset} \{\empk \}V(q)\{\empk\},$$
where 
\begin{itemize}
\item $P(q) \equiv\withk{se}{\textsf{p}=0}{\assigk{q}{1}}$ and
\item $V(q)\equiv  \withkk{se}{\assigk{q}{0}}$.
\end{itemize}

Using the $(PAR)$ rule, we obtain the next specification.

$$se(\textsf{p},\textsf{q}): S \vdash_{\emptyset} \{\empk \} \park{(\seqk{P(p)}{V(p)})}{(\seqk{P(q)}{V(q)})} \{\empk \}$$

This program is a solution to the critical region problem. Next, we add a Critical Region between the operation $P$ and $Q$ in the previous specification. And, we discuss
the properties of mutual exclusion, absence of deadlocks and starvation.

Consider a Critical Region (C.R.) and the following program
\begin{center}
\begin{tabular}{c}
\seqk{\assigk{p}{0}}{\assigk{q}{0}};\\
\resk{se}{}\\
\begin{tabular}{c||c}
\whik{true}{}& \whik{true}{}\\
P(p); & P(q);  \\
C.R.; & C.R.;\\
V(p) & V(q)
\end{tabular}
\end{tabular}
\end{center}

The execution of the program is inside the critical region for the thread $p$ ($q$), if $p=1$ ($q=1$, respectively). The mutual exclusion follows from the resource invariant $S$.

The execution of this program is free from deadlock, because the resource invariant implies that one of the control variables $\textsf{p}=0 \vee\textsf{q}=0$.

After the execution of the critical region, the execution of $V(p)$ or $V(q)$  allows the execution of $P(p)$ and $P(q)$. Assuming the fairness of the parallel execution, the program is free from starvation.

\subsection{Concurrent stack}
First, to show that DCSL is not as expressive as CSL, we present an
example of parallel operations over a stack that cannot be proved
correct in the former but can be in the latter.

Let us specify a stack with operations \textit{pop} and
\textit{push}. The stack is represented by the resource $st$ in the
following way
$$st(\{ \textsf{z}, \textsf{y}\}): stack(\textsf{z}),$$
where $\{ \textsf{z}, \textsf{y}\}$ is the set of variables protected
by $st$, and $stack(\textsf{z})$ is defined by
$$(\textsf{z}=null\wedge \empk) \vee (\exists_{a,b} \textsf{z}\mapsto a,b \ast stack(b)).$$

The operations \textit{pop} and \textit{push} over a stack are defined below.
$$pop(\textsf{x1})\equiv\withk{st}{\neg(\textsf{z} = null)}{(\seqk{\seqk{\seqk{\assigk{y}{z}}{\assigk{x1}{y}}}{\loadk{z}{y+1}}}{\dispk{y+1}})},$$
$$push(\textsf{x2})\equiv\withkk{st}{(\consk{y}{x2,z};\assigk{z}{y}\hfill)}.$$

The operation \textit{pop} picks the first node of a non-empty stack
and passes it to the variable $\textsf{x1}$.  In the following specification,
the program performs a pop over a shared stack and it disposes the
memory space retrieved by the stack.
\begin{equation}\label{spepop}
st(\textsf{z},\textsf{y}): stack(z)\vdash \{\empk\} \seqk{pop(\textsf{x1})}{\dispk{x1}} \{ \empk\}.
\end{equation}

To prove this result in DCSL, we use the rules of SL and the critical
region rule, by omitting
the rely-set. Consider the following derivation, that proves the
validity of the program inside the critical region.
\begin{eqnarray*}
&\vdash&\{\empk \ast \exists_{a,b}\  \textsf{z}\mapsto a,b \ast stack(b)\}\\&&
\assigk{y}{z};\\&&
\{\empk\ast \exists_{a,b}\  \textsf{y}\mapsto a,b \ast stack(b)\}\\&&
 \assigk{x1}{y};\\&&
\{ \empk \ast \exists_{a,b}\  \textsf{x1}\mapsto a\ast \textsf{y+1}\mapsto b \ast stack(b))\}\\&&
 \loadk{z}{y+1};\\&&
\{ \empk \ast \exists_{a}\  \textsf{x1}\mapsto a\ast \textsf{y+1}\mapsto \textsf{z} \ast stack(\textsf{z}))\}\\&&
 \dispk{\textsf{y+1}}\\&&
\{ (\exists_{a}\  \textsf{x1}\mapsto a) \ast (stack(\textsf{z}))\}
\end{eqnarray*}

Applying the critical region rule, the resource $st$ appears and we obtain, 
$$st(\textsf{z},\textsf{y}): stack(\textsf{z})\vdash \{\empk\}pop(\textsf{x1})\{\ \exists_{a}\  \textsf{x1}\mapsto a\}.$$
Using the sequential and deallocation rules of SL we get the
specification $(\ref{spepop})$.

Now, we turn our attention to the \textit{push} operator over a stack,
showing that the following specification is valid in the context of
DCSL. Let \textit{push} insert an element $x2$ in the top of a stack.
$$st(\textsf{z},\textsf{y}): stack(\textsf{z})\vdash \{\empk \}push(\textsf{x2})\{\empk\}.$$
As before, from SL inference rules, we obtain the specification below.
Then we can apply the critical region rule to obtain the specification
above.
\begin{eqnarray*}
&\vdash&\{\empk \ast stack(\textsf{z})\}\\&&
\consk{y}{x2,z};\\&&
\{\textsf{y}\mapsto \textsf{x2},\textsf{z} \ast stack(\textsf{z})\}\\&&
\{\empk \ast \exists_{a,b}\  \textsf{y}\mapsto a,b \ast stack(b)\}\\&&
\assigk{z}{y}\\&&
\{\empk \ast \exists_{a,b} \ \textsf{z}\mapsto a,b \ast stack(b)\}
\end{eqnarray*}

Until now we have shown that each specification is derivable in DCSL; now
we want to study their parallel composition. To apply the parallel
rule we need that
the variables modified by one program cannot occur free in the other.
$$mod(\seqk{pop(\textsf{x1})}{\dispk{x1}})=\{\textsf{x1},\textsf{y},\textsf{z}\},\quad mod(push(\textsf{x2}))=\{\textsf{y}, \textsf{z}\}.$$

The variables $\textsf{z}$ and $\textsf{y}$ are used in both
specifications. Hence it is not possible to apply the DCSL parallel
rule and obtain a specification for the parallel execution of
\textit{pop} and \textit{push}.

In order to express the specification above in the context of CSL it
is necessary to define the rely-set for the operation of \textit{pop}
and \textit{push}, that are ${\{\textsf{x1}\}}$ and
${\{\textsf{x2}\}}$, respectively.

It is straightforward, using the derivations above, to infer, in CSL,
the following specifications:
$$st(\textsf{z},\textsf{y}): stack(\textsf{z})\vdash_{{\{\textsf{x1}\}}} \{\empk\} \seqk{pop(\textsf{x1})}{\dispk{x1}} \{ \empk\},$$
$$st(\textsf{z},\textsf{y}): stack(\textsf{z})\vdash_{\{\textsf{x2}\}} \{\empk \}push(\textsf{x2})\{\empk\}.$$

To apply the CSL parallel rule, we need to check that there is no
interference between rely-sets and modified variables.  Since
$mod(push(\textsf{x2}))\cap \{\textsf{x1}\}=\emptyset$ and
$mod(\seqk{pop(\textsf{x1})}{\dispk{x1}})\cap \{\textsf{x2}\}=\emptyset$, by
parallel rule we infer:
$$st(\textsf{z},\textsf{y}):
  stack(\textsf{z})\vdash_{\{\textsf{x1},\textsf{x2}\}}\{\empk \}\park{(\seqk{pop(\textsf{x1})}{\dispk{x1}})}{push(\textsf{x2})}\{\empk\}.$$
As this example shows we can obtain, at least, simpler specifications
using CSL than DCSL, and prove correctness of more programs.

\section{Operational Semantics}

In this section, we describe a structural operational semantics (SOS) that we use to prove the soundness of CSL.
We mostly follow the approach of Vafeiadis \cite{vaf}. Let us introduce the concept of resource configuration, which records the state of each declared resource.  

\subsection{Program transition}\label{progtrans}

We start by extending the programming language with a command for executions inside a critical region. We denote this command by $\withink{r}{C}$, where $r$ is an acquired resource and $C$ is a command. In the extended programming language, we can associate to each command a set of \emph{locked resources}, $Locked(C)$ which is inductively defined by:
\begin{align*}
&Locked(\seqk{C_{1}}{C_{2}})=Locked(C_1),\\
 &Locked(\park{C_1}{C_2})=Locked(C_1)\cup Locked(C_2),\\
 &Locked(\resk{r}{C})=Locked(C)\setminus \{r\},\\
&Locked(\withink{r}{C})=Locked(C)\cup \{r\},\\
 &Locked(C)=\emptyset\textit{, otherwise.}
\end{align*}
Moreover the set of resources occurring in a command is extended with all resources names $r$ such that the command also includes $\withink{r}{C}$.

Let $O,L,D$ be disjoint pairwise subsets of resources names. We say that $\rho=(O,L,D)$ is a \emph{resource configuration}, where $O$ are resources owned by the running program, $L$ are resources locked by others programs and $D$ are available resources. The set of resources configurations is denoted by $\mathcal{O}$.

 We write $r\in \rho$ ($r\notin \rho$) if $r\in (O\cup L\cup D)$ ($r\notin (O\cup L\cup D)$, respectively), and $(O, L, D)\setminus \{r\}:=(O\setminus \{r\}, L\setminus \{r\}, D\setminus \{r\})$. 

Usually the state of a machine in SL consists of a storage, $s$, and a heap, $h$. However, we define a program's state by a triple $(s,h,\rho)$.
The program transitions, that define the SOS, are represented by the relation $\rightarrow_p$ defined from the tuple $(C,(s,h,\rho))$ to $(C',(s',h',\rho'))$ or the abort state ($\abok$), where $C, C'\in \mathcal{C}$, $(s,h),(s',h')\in \mathcal{S}$ and $\rho,\rho'\in\mathcal{O}$.

\begin{figure*}
\resizebox{\textwidth}{\height}{
$$
\infer[(S1)]{\skipk;C_2, (s,h,\rho)\rightarrow_{p} C_2, (s,h,\rho) }{}
\quad
\infer[(S2)]{\seqk{C_{1}}{C_{2}}, (s,h,\rho)\rightarrow_{p} \seqk{C'_1}{C_2}, (s',h',\rho') }{C_1, (s,h,\rho)\rightarrow_{p} C'_1, (s',h',\rho')}
$$}
$$
\infer[(LP)]{\whik{B}{C},(s,h,\rho)\rightarrow_{p} \ifk{B}{\seqk{C}{\whik{B}{C}}}{\skipk},(s,h,\rho)}{}
$$
\resizebox{\textwidth}{\height}{
$$\infer[(IF1)]{\ifk{B}{C_1}{C_2}, (s,h,\rho)\rightarrow_{p} C_1, (s,h,\rho)}{s(B)=\texttt{true}}
\quad \infer[(P1)]{\park{C_1}{C_2}, (s,h,\rho)\rightarrow_{p} \park{C'_1}{C_2}, (s',h',\rho') }{C_1, (s,h,\rho)\rightarrow_{p} C'_1, (s',h',\rho')}$$}
\vspace{0.3mm}

\resizebox{\textwidth}{\height}{
$$
\infer[(IF2)]{\ifk{B}{C_1}{C_2}, (s,h,\rho)\rightarrow_{p} C_2, (s,h,\rho)}{s(B)=\texttt{false}}
\quad
\infer[(P2)]{\park{C_1}{C_2}, (s,h,\rho)\rightarrow_{p} \park{C_1}{C'_2}, (s',h',\rho')}{C_2, (s,h,\rho)\rightarrow_{p} C'_2, (s',h',\rho')}
$$}

\resizebox{\textwidth}{\height}{
$$
\infer[(R0)]{\resk{r}{\skipk}, (s,h,\rho)\rightarrow_{p} \skipk, (s ,h,\rho) }{r\notin \rho }
\quad
\infer[(P3)]{\skipk \| \skipk, (s,h,\rho)\rightarrow_{p} \skipk, (s,h,\rho)}{}
$$}
\vspace{0.5mm}

\resizebox{\textwidth}{\height}{
$$
\infer[(R1)]{\resk{r}{C}, (s,h,\rho)\rightarrow_{p} \resk{r}{C}', (s',h',\rho'\setminus\{r\})}{ C, (s,h,(O\cup \{r\},L,D))\rightarrow_{p} C', (s',h',\rho') & r\notin \rho=(O,L,D)  & r\in Locked(C)}
$$}
\vspace{0.5mm}

\resizebox{\textwidth}{\height}{
$$
\infer[(R2)]{\resk{r}{C}, (s,h,\rho)\rightarrow_{p} \resk{r}{C}', (s',h',\rho'\setminus\{r\})}{C, (s,h,(O,L,D\cup \{r\}))\rightarrow_{p} C', (s',h',\rho') & r\notin \rho=(O,L,D) &  r\notin Locked(C)}
$$}
$$
\infer[(W0)]{\withk{r}{B}{C}, (s,h,\rho)\rightarrow_{p} \withink{r}{C}, (s,h,\rho')}{ \rho=(O,L,D\cup\{r\}) & \rho'=(O\cup\{r\},L,D) & s(B)=\texttt{true}}
$$
\resizebox{\textwidth}{\height}{$$
\infer[(W1)]{\withink{r}{C}, (s,h,(O,L,D))\rightarrow_{p} \withink{r}{C}', (s',h',(O'\cup\{r\},L',D')}{r\in O & C, (s,h,(O\setminus \{r\},L,D))\rightarrow_{p} C', (s',h',(O',L',D'))}
$$}
\vspace{0.5mm}

\resizebox{\textwidth}{\height}{
$$\infer[(W2)]{\withink{r}{\skipk}, (s,h,\rho)\rightarrow_{p} \skipk, (s,h,\rho')}{\rho=(O\cup\{r\},L,D) & \rho'=(O,L,D\cup\{r\})}
\quad
\infer[(BCT)]{c, (s,h,\rho)\rightarrow_{p} \skipk, (s',h',\rho)}{[c](s,h)=(s',h')}
$$}

\caption{Program Transitions}
\label{fig:op}
\end{figure*}
For a basic command $c$ we denote by $[c](s,h)$ the result of executing $c$ for the pair $(s,h)$, in the context of SL. The result of the execution of $c$ on a pair $(s,h)$ can be a pair $(s',h')$ or $\abok$.
In Figure~\ref{fig:op}, we display the program transitions.

Since most of the program transitions are standard, we only emphasize how we manage the resource configuration. First note that it is not changed by any transition of basic commands $(BCT)$. The acquisition of a resource by the transition $(W0)$ requires that the resource is available and transfers it to the set of owned resources; the release of a resource made by $(W2)$ returns the resource to the set of available resources. The local resource command does not add the resource to the resource configuration, since that would break locality, i.e., the local resource should only be visible to who created it. For the local resource we use the set of locked resources, $Locked(C)$, to determine if a resource should be in the set of owned or available resources.  In Proposition~\ref{lockres}, we prove that $Locked(C)$ is equal to the set of owned resources along an execution starting in a non-extended command.

\begin{figure*}
\resizebox{\textwidth}{\height}{
$$
\infer[(RA)]{\resk{r}{C}, (s,h,\rho)\rightarrow_{p} \abok }{r\in \rho}
\quad
\infer[(WA)]{\withk{r}{B}{C}, (s,h,\rho)\rightarrow_{p} \abok }{r\notin \rho}
$$}
\vspace{0.5mm}

\resizebox{\textwidth}{\height}{
$$
\infer[(RA1)]{\resk{r}{C}, (s,h,(O,L,D))\rightarrow_{p} \abok}{r\in Locked(C) & C, (s,h,(O\cup \{r\},L,D) )\rightarrow_{p} \abok}
\quad
\infer[(BCA)]{c, (s,h,\rho)\rightarrow_{p} \abok}{[c](s,h)=\abok}
$$}
\vspace{0.5mm}

\resizebox{\textwidth}{\height}{
$$
\infer[(RA2)]{\resk{r}{C}, (s,h,(O,L,D))\rightarrow_{p} \abok}{r\notin Locked(C) & C, (s,h,(O,L,D\cup \{r\}) )\rightarrow_{p} \abok}
\quad
\infer[(SA)]{\seqk{C_{1}}{C_{2}}, (s,h,\rho)\rightarrow_{p} \abok }{C_1,(s,h,\rho)\rightarrow_{p} \abok}
$$}
\vspace{0.5mm}

\resizebox{\textwidth}{\height}{
$$
\infer[(WA1)]{\withink{r}{C}, (s,h,\rho)\rightarrow_{p} \abok}{C, (s,h,\rho\setminus\{r\})\rightarrow_{p} \abok}
\quad
\infer[(PA1)]{\park{C_1}{C_2}, (s,h,\rho)\rightarrow_{p} \abok }{C_1, (s,h,\rho)\rightarrow_{p} \abok}
$$}
\vspace{0.5mm}

\resizebox{\textwidth}{\height}{
$$
\infer[(WA2)]{\withink{r}{C}, (s,h,(O,L,D))\rightarrow_{p} \abok }{r\notin O}
\quad
\infer[(PA2)]{\park{C_1}{C_2}, (s,h,\rho)\rightarrow_{p} \abok }{C_2, (s,h,\rho)\rightarrow_{p} \abok}
$$}
\caption{Abort Transitions}
\label{fig:abort}
\end{figure*}

In Figure~\ref{fig:abort}, we include transitions that abort. As in SL, a memory fault causes the program to abort. The parallel command aborts if one of its commands aborts. The local resource command aborts, if the command tries to create a pre-existing resource. The critical region command aborts if it tries to acquire an undeclared resource, if the execution inside the critical region aborts, or if an acquired resource is not in the set of owned resources.

Next, we check that program transitions are well-defined.

\begin{prop}\label{proprec}
Let $C,C'	\in\mathcal{C}$, $(s,h),(s',h')\in \mathcal{S}$ and $O,L,D, O',L',D'\subset \textbf{Res}$. If $(O,L,D)\in\mathcal{O}$ and $$C,(s,h,(O,L,D))\rightarrow_{p} C',(s',h',(O',L',D')),$$ then  $(O',L',D')\in \mathcal{O}$, $L=L'$ and $O\cup D=O'\cup D'$.
\end{prop}

\begin{proof}
We prove the result by induction on the rules of $\rightarrow_{p}$.

Let $C,C'	\in\mathcal{C}$, $(s,h),(s',h')\in \mathcal{S}$, $O,O',L,L',D,D'\subseteq \textbf{Res}$ such that $(O,L,D)\in\mathcal{O}$ and $$C,(s,h,(O,L,D))\rightarrow_{p} C',(s',h',(O',L',D')).$$

The proof is immediate for the rules $(S1)$, $(S2)$, $(LP)$, $(IF1)$, $(IF2)$, $(P1)$, $(P2)$, $(P3)$, $(R0)$ and $(BCT)$.

If the transition is $(R1)$. 
We have that $r\notin(O,L,D)$, $C=\resk{r}{\tilde{C}}$, $C'=\resk{r}{\tilde{C'}}$ and $r\in Locked(\tilde{C})$ such that  $$\tilde{C},(s,h,(O\cup \{r\},L,D))\rightarrow_{p} \tilde{C'},(s',h',(O'',L'',D'')),$$ where $(O'',L'',D'')\setminus \{r\} =(O',L',D')$.

By induction $L=L''$, $(O'',L'',D'')\in \mathcal{O}$ and $O\cup \{r\}\cup D=O''\cup D''$. 

From $r\notin L$, it follows that $r\notin L''$. Hence $L=L'$.
Because $r\notin (O,L,D)$, 
$$O'\cup D'= (O''\cup D'')\setminus \{r\} = (O\cup \{r\}\cup D)\setminus \{r\} = O\cup D.$$

From $O'=O''\setminus \{r\}$ and $D'=D''\setminus \{r\}$, we obtain that 
$$O'\cap D'= (O''\setminus \{r\})\cap(D''\setminus \{r\})\subseteq O''\cap D''=\emptyset,$$
$$O'\cap L' = (O''\setminus \{r\})\cap L''\subseteq O''\cap L''= \emptyset,$$
$$D'\cap L' = (D''\setminus \{r\})\cap L''\subseteq D''\cap L''= \emptyset.$$

Hence $(O',L',D')\in \mathcal{O}$.

The case $(R2)$ is analogous to the previous one.

If the transition is given by $(W0)$.
We have $D'=D\setminus \{r\}$, $O'=O\cup \{r\}$, $L=L'$ and $r\in D$. And 
$$O'\cap D'=(O\cup \{r\})\cap (D\setminus \{r\})= (O\cap (D\setminus \{r\}))\subseteq O\cap D = \emptyset,$$
$$O'\cap L' = (O\cup \{r\})\cap L= (O\cap L) \cup( \{r\}\cap L)= \{r\}\cap L=\emptyset,$$
$$L'\cap D' = L \cap  (D\setminus \{r\}) \subseteq L\cap D=\emptyset.$$

Therefore $(O',L',D')\in \mathcal{O}$. Moreover,  $O'\cup D'= (O\cup \{r\}) \cup (D\setminus \{r\})= O\cup D$ and $L=L'$.

If the transition is given by $(W1)$.
We have $r\in O$, $C= \withink{r}{\tilde{C}}$ and $C'=\withink{r}{\tilde{C'}}$  such that $$\tilde{C},(s,h,(O\setminus \{r\},L,D))\rightarrow_{p} \tilde{C'},(s',h',(O'',L'',D'')),$$ where  $O'=O''\cup \{r\}$, $L'=L''$, $D'=D''$.

By induction hypothesis, we know that $(O'',L'',D'')\in \mathcal{O}$, $L=L''$ and $(O\setminus \{r\})\cup D=O''\cup D''$.

Therefore $L=L'$ and $O\cup D= \{r\}\cup O''\cup D''= O'\cup D'$.

It remains to check that $(O',L',D')\in \mathcal{O}$, this follows from
$$O'\cap L'=(O''\cup \{r\}) \cap L'' = (O''\cap L'')\cup (\{r\}\cap L) =\emptyset,$$
$$O'\cap D' =  (O''\cup \{r\}) \cap D'' = (O''\cap D'')\cup (\{r\}\cap D'')=\emptyset,$$
$$L'\cap D' = L''\cap D''= \emptyset.$$

If the transition is given by $(W2)$.
We have $r\in O$, $L=L'$, $D'=D\cup \{r\}$ and $O'=O\setminus \{r\}$. Note that
$$O'\cap L'\subseteq O\cap L = \emptyset,$$
$$O'\cap D' =  (O\setminus \{r\}) \cap (D\cup \{r\}) \subseteq (O\cap D) \cup ((O\setminus \{r\}) \cap \{r\})=\emptyset,$$
$$L'\cap D' = L\cap (D\cup \{r\})= (L\cap D) \cup (L\cap \{r\})=\emptyset.$$

Then $(O',L',D')\in \mathcal{O}$, $L=L'$ and $O'\cup D'= (O\setminus \{r\}) \cup (D\cup \{r\})= O\cup D$.
\end{proof}

We say that a command $C'$ is \emph{reachable} from a CSL's command $C$ if there are $(s,h,\rho), (s',h',\rho')$ and $k$ such that $$C,(s,h,\rho)\rightarrow_p^k C',(s',h',\rho')$$ and $C,(s,h,\rho)\not\rightarrow_p^j \abok$ for every $j\leq k$, where $\rightarrow_p^i$ denotes the composition of $i$ transitions. In the next proposition, we see that owned resources are equal to locked resources, along an execution starting from a non-extended command.

\begin{prop}\label{lockres}
Let $C,C'\in \mathcal{C}$, $(s,h),(s',h')\in \mathcal{S}$, $\rho'\in \mathcal{O}$, $\Gamma$ be a resource context, and $k\geq 0$ such that $C'$ is reachable from $C$. 

If $C,(s,h,(\emptyset,\emptyset,Res(\Gamma))\rightarrow_p^k C',(s',h',\rho')$, then $$\rho'=(Locked(C'), \emptyset, Res(\Gamma)\setminus Locked(C')).$$
\end{prop}

\begin{proof}
Let $C,C'\in \mathcal{C}$, $(s,h),(s',h')\in \mathcal{S}$, $\rho'\in \mathcal{O}$, $\Gamma$ be a resource context, and $k\geq 0$ such that $C'$ is reachable from $C$. 

The proof is done by  induction on $k$.

Let $k=0$. Then $C'=C$ is a non-extended command and $Locked(C)=\emptyset$.

Let $k=n+1$. Then there exist $C''$, $s''$, $h''$, $\rho''$ such that 
$$C,(s,h,(\emptyset,\emptyset,Res(\Gamma))\rightarrow_p^n C'',(s'',h'',\rho'')\rightarrow_p C',(s',h',\rho').$$

Note that $C''$ is reachable from $C$. By the induction hypothesis on $k$, we have that $$\rho''=(Locked(C''), \emptyset, Res(\Gamma)\setminus Locked(C'')).$$

Now, we prove the result from $k$ to $k+1$ by induction on the program transitions. The proof is immediate or an immediate consequence of the induction on the program transition for all transitions except $(W0)$ and $(W2)$, which change the locked resources.

However in both cases the resulting resource configuration preserves that \[\rho'=(Locked(C'), \emptyset, Res(\Gamma)\setminus Locked(C')).\qedhere \]

\end{proof}

The proposition above reinforces the idea that the transitions $(R1)$ and $(R2)$ are well defined. Furthermore, it completely describes the resource configuration along an execution. 


\subsection{Properties of program transitions}\label{smfp}

We state now the main properties of the program transitions.
We start with the safety monotonicity and the frame property that are essential to show the soundness of the frame rule, as well as of the parallel rule. The property of safety monotonicity and frame property original appears in the context of Separation Logic and are still valid in the Concurrent Separation Logic.

Let $h,g$ be heaps. If they have disjoint domains we write $h\bot g$, and we denote by $h\uplus g$ the union of disjoint heaps. If $g$ is a subheap of $h$,  $h\setminus g$ denotes the heap $h$ restricted to $dom(h)\setminus dom(g)$.

\begin{prop}\label{safetymono}
Let $C\in \mathcal{C}$, $(s,h)\in \mathcal{S}$, and $\rho\in\mathcal{O}$. Suppose $h_F$ is a heap such that $h\bot h_F$.
If $C,(s,h,\rho)\not\rightarrow_{p} \abok$, then $C,(s,h\uplus h_F,\rho)\not\rightarrow_{p} \abok$.
\end{prop}

\begin{prop}\label{frameprop}
Let  $C,C'\in \mathcal{C}$, $(s,h), (s',h')\in\mathcal{S}$, and $\rho,\rho'\in\mathcal{O}$.

Suppose $h_F$ is a heap such that $h\bot h_F$. If $C,(s,h\uplus h_F,\rho)\rightarrow_{p} C',(s',h',\rho')$ and $C,(s,h,\rho)\not\rightarrow_{p} \abok$, then $h_F$ is a subheap of $h'$ and 
$$C,(s,h,\rho)\rightarrow_{p} C',(s',h'\setminus h_F,\rho').$$
\end{prop}

The proofs follow a standard pattern (See \ref{Appendix} for the proofs).

By safety monotonicity and frame property we know that the execution of parallel commands only affects his own heap; however it is necessary to have dual properties for the resource configuration. Next, we state these dual properties.

\begin{prop}\label{abortpar}
Let $C\in \mathcal{C}$, $(s,h)\in \mathcal{S}$, and $(O_1\cup O_2,L, D),(O_1,L\cup O_2, D)\in \mathcal{O}$.
If $C,(s,h,(O_1,L\cup O_2, D))\not\rightarrow_{p} \abok$, then $$C,(s,h,(O_1\cup O_2,L, D))\not\rightarrow_{p} \abok.$$
\end{prop}

\begin{proof}
We will prove the contra-position by induction on the rules of $\rightarrow_{p}$.

Let $C\in \mathcal{C}$, $(s,h)\in \mathcal{S}$, and $(O_1\cup O_2,L, D),(O_1,L\cup O_2, D)\in \mathcal{O}$ such that 
$$
C,(s,h,(O_1\cup O_2,L, D))\rightarrow_{p} \abok.
$$

If the transition is given by $(BCA)$.
The transition is independent from the resource configuration. Then 
$$C,(s,h,(O_1,L\cup O_2, D))\rightarrow_{p} \abok.$$

If the transition is given by $(RA)$ or $(WA)$.
The conclusion follows from $$r\in(O_1,L\cup O_2, D)\ \  \textrm{iff} \ \ r\in(O_1\cup O_2,L, D).$$

If the transition is given by $(WA2)$.
We have $r\notin O_1\cup O_2$. Then $r\notin O_1$ and $$C,(s,h,(O_1,L\cup O_2, D))\rightarrow_{p} \abok.$$

If the transition is given by $(PA1)$.
We have  $$C_1,(s,h,(O_1\cup O_2,L, D))\rightarrow_{p} \abok.$$ 

Using the induction hypotheses, $C_1,(s,h,(O_1,L\cup O_2, D))\rightarrow_{p} \abok$. Then $$C,(s,h,(O_1,L\cup O_2, D))\rightarrow_{p} \abok.$$

The cases $(PA2)$, $(SA)$, $(RA1)$, $(RA2)$ and $(WA1)$ are similar to the previous case.
\end{proof}

The dual of the frame property for resource configurations is the following.

\begin{prop}\label{p1}
Let $C,C'\in\mathcal{S}$, $(s,h),(s',h')\in\mathcal{S}$, and $\rho_1,\rho_2,\rho'\in \mathcal{O}$ such that $\rho'=(O',L, D')$, $\rho_1=(O_1\cup O_2,L, D)$ and $\rho_2=(O_1,L\cup O_2, D)$.
If $C,(s,h,\rho_2)\not\rightarrow_{p} \abok$ and  $C, (s,h,\rho_1)\rightarrow_{p} C', (s',h',\rho')$, then $O_2\subseteq O'$ and $$C, (s,h,\rho_2)\rightarrow_{p} C', (s',h',(O'\setminus O_2,L\cup O_2, D')).$$
\end{prop}

\begin{proof}

Let $C,C',s,s',h,h',O_1,L,O_2,D, O',D'$ as stated in the proposition. 

Suppose that  $C,(s,h,(O_1,L\cup O_2,D))\not\rightarrow_{p} \abok$ and $$C, (s,h,(O_1\cup O_2,L, D))\rightarrow_{p} C', (s',h',(O',L, D')).$$

The prove is done by induction on the program transitions.

If the transition is given by $(S1)$, $(LP)$, $(IF1)$, $(IF2)$, $(BCT)$ or $(P3)$, the transition does not depend on the resource configuration. Then, the conclusion is immediate.

If the transition is given by $(W0)$.

We have $C=\withk{r}{B}{\tilde{C}}$, $C'=\withink{r}{\tilde{C}}$, $r\in D$, $s(B)=\texttt{true}$, $s'=s$, $h'=h$, $O'=O_1\cup O_2 \cup  \{r\}$ and $D'=D\setminus \{r\}$.

 Then $O_2\subseteq O'$ and $O'\setminus O_2=O_1\cup \{r\}$. Therefore
$$C, (s,h,(O_1,L\cup O_2, D))\rightarrow_{p} C', (s',h',(O'\setminus O_2,L\cup O_2, D')).$$

If the transition is given by $(W1)$.

We have $C=\withink{r}{\tilde{C}}$, $C'=\withink{r}{\tilde{C'}}$, $r\in (O_1\cup O_2)\cap O'$ and $$\tilde{C}, (s,h,(O_1\cup O_2\setminus \{r\},L, D))\rightarrow_{p} \tilde{C'}, (s',h',(O'\setminus \{r\},L, D')).$$

From $C,(s,h,(O_1,L\cup O_2, D))\not\rightarrow_{p} \abok$, we know that $r\in O_1$ and $$\tilde{C}, (s,h,(O_1\setminus \{r\},L\cup O_2, D))\not\rightarrow_{p} \abok.$$

Now, we can apply the induction hypothesis to conclude that $O_2\subseteq O'\setminus \{r\}$ and $$\tilde{C}, (s,h,(O_1\setminus \{r\},L\cup O_2, D))\rightarrow_{p} \tilde{C'}, (s',h',((O'\setminus \{r\})\setminus O_2,L\cup O_2, D')).$$

Note that $O_2\subseteq O'\setminus \{r\}\subseteq O'$ and $(O'\setminus \{r\})\setminus O_2=(O'\setminus O_2)\setminus \{r\}$.

From $r\in O_1\cap O'$ and $r\notin O_2$, we know that $r\in O_1 \cap (O'\setminus O_2)$.
Therefore $$C, (s,h,(O_1,L\cup O_2, D))\rightarrow_{p} C', (s',h',(O'\setminus O_2,L\cup O_2, D')).$$

If the transition is given by $(W2)$.

We have $C=\withink{r}{\skipk}$, $C'=\skipk$, $s'=s$, $h'=h$, $r\in O_1\cup O_2$, $O'=(O_1\cup O_2) \setminus \{r\}$ and $D'=D\cup \{r\}$.

As before, we know that $r\in O_1$. Hence we can rewrite the set of owned resources in the following expression $$O'=(O_1\setminus \{r\}) \cup O_2.$$

Then $O_2\subseteq O' $ and $$C, (s,h,(O_1,L\cup O_2, D))\rightarrow_{p} C', (s',h',(O'\setminus O_2,L\cup O_2, D')).$$

If the transition is given by $(R0)$.

We have $C=\resk{r}{\skipk}$, $C'=\skipk$, $s'=s$, $h'=h$, $r\notin(O_1\cup O_2,L, D)$, $O'=O_1\cup O_2$ and $D'=D$. Then $O_2\subseteq O'$. 

From $r\notin (O_1\cup O_2,L, D)$, we know that $r\notin (O_1,L\cup O_2, D)$. Therefore $$C, (s,h,(O_1,L\cup O_2, D))\rightarrow_{p} C', (s',h',(O'\setminus O_2,L\cup O_2, D')).$$

If the transition is given by $(R1)$.

We have $C=\resk{r}{\tilde{C}}$, $C'=\resk{r}{\tilde{C'}}$, $r\notin (O_1\cup O_2,L, D)$, $r\in Locked(\tilde{C})$ and
$$\tilde{C}, (s,h,(O_1\cup O_2\cup\{r\},L, D))\rightarrow_{p} \tilde{C'}, (s',h',(O'',L, D'')),$$ such that $O''\cup D''=O'\cup D' \cup \{r\}$.

From $C, (s,h,(O_1,L\cup O_2, D))\not\rightarrow_{p} \abok$, we know that $r\notin(O_1,L\cup O_2, D)$ and $$\tilde{C}, (s,h,(O_1\cup\{r\},L\cup O_2, D))\not\rightarrow_{p} \abok.$$

By induction hypothesis, we have that $O_2\subseteq O''$ and $$\tilde{C},(s,h,(O_1\cup \{r\}, L\cup O_2, D))\rightarrow_{p} \tilde{C},(s,h,(O''\setminus O_2, L\cup O_2, D'')).$$

From $O''\subseteq O'\cup\{r\}$ and $r\notin O_2$, we have that $O_2\subseteq O'$.

Moreover, we get that $(O''\setminus O_2 )\cup D''= (O'\setminus O_2) \cup D'\cup\{r\}$. Therefore $$C, (s,h,(O_1,L\cup O_2, D))\rightarrow_{p} C', (s',h',(O'\setminus O_2,L\cup O_2, D')).$$

The case $(R2)$ is similar to the previous case.

If the transition is given by $(P1)$.

We have $C=\park{C_{1}}{C_{2}}$, $C'= \park{C'_{1}}{C_{2}}$ and $$C_1,(s,h,(O_1\cup O_2,L,D))\rightarrow_{p} C'_1,(s',h',(O',L,D')).$$ 
And $C_1,(s,h,(O_1,L\cup O_2,D))\not\rightarrow_{p} \abok$, since $C,(s,h,(O_1,L\cup O_2,D))\not\rightarrow_{p} \abok$.

By the induction hypothesis, we conclude that $O_2\subseteq O'$ and $$C_1,(s,h,(O_1,L\cup O_2,D))\rightarrow_{p} C'_1,(s',h',(O'\setminus O_2,L\cup O_2,D')).$$

Therefore $$C, (s,h,(O_1,L\cup O_2, D))\rightarrow_{p} C', (s',h',(O'\setminus O_2,L\cup O_2, D')).$$

The cases $(P2)$ and $(S2)$ are similar to the previous case.
\end{proof}

The previous propositions allow us to make a correspondence between the transitions in a parallel execution to transitions of its commands executed independently.

For the soundness of the renaming rule, we prove that the execution of a program and its renaming version are equivalents.

\begin{prop}\label{abortrename}
Let $C,C'\in \mathcal{C}$, $(s,h), (s',h')\in\mathcal{S}$, $\rho,\rho'\in\mathcal{O}$, and $r,r'\in\textbf{Res}$ such that $r'\notin Res(C)$ and $r'\notin \rho$.
\begin{enumerate}
\item  $C, (s,h,\rho)\not\rightarrow_{p} \abok$ if and only if $C[r'/r], (s,h,\rho [r'/r])\not\rightarrow_{p} \abok$.

\item  $C, (s,h,\rho)\rightarrow_{p} C', (s',h',\rho')$ if and only if $$C[r'/r], (s,h,\rho[r'/r])\rightarrow_{p} C'[r'/r], (s',h',\rho'[r'/r]).$$
\end{enumerate}
\end{prop}


\begin{proof}
Let $C,C'\in \mathcal{C}$, $(s,h), (s',h')\in\mathcal{S}$, $\rho,\rho'\in\mathcal{O}$, and $r,r'\in\textbf{Res}$ such that $r'\notin Res(C)$ and $r'\notin \rho$.

Note that $C[r'/r][r/r']=C$ and $\rho [r'/r][r/r']=\rho$. Hence, we only need to prove one direction of the equivalence. We prove both by induction on $\rightarrow_{p}$.

For the first one, we suppose that $C, (s,h,\rho)\rightarrow_{p} \abok$ and prove that $C[r'/r], (s,h,\rho [r'/r])\rightarrow_{p} \abok$.

Suppose that the transition to the $\abok$ state is given by $(BCA)$.

Note that the transitions is independent from the resource configuration and $C[r'/r]=C$.
Then $$C[r'/r], (s,h,\rho[r'/r])\rightarrow_{p} \abok.$$

Suppose that the transition is given by $(PA1)$. Then $C=\park{C_{1}}{C_{2}}$ and $$C_1, (s,h,\rho)\rightarrow_{p} \abok.$$

We know that $r'\notin Res(C_1)$, because $r'\notin Res(C)$.
Using the induction hypothesis, we have that $$C_1[r'/r], (s,h,\rho[r'/r])\rightarrow_{p} \abok.$$ 

Hence $$C[r'/r], (s,h,\rho [r'/r])\rightarrow_{p} \abok.$$

The cases $(P2)$ and $(SA)$ are identical to the previous case.

Suppose that the transition is given by $(RA)$. We have $C=\resk{\hat{r}}{\tilde{C}}$ and $\hat{r}\in\rho$.

Note that $C[r'/r]=\resk{\hat{r}[r'/r]}{\tilde{C}[r'/r]}$ and $$\hat{r}\in\rho\ \ \textrm{iff}\ \ \hat{r}[r'/r]\in\rho[r'/r].$$

Therefore $$C[r'/r], (s,h,\rho[r'/r])\rightarrow_{p} \abok.$$

The cases $(WA)$ and $(WA2)$ are analogous to the case before.

Suppose that the transition is given by $(RA1)$. We have $C=\resk{\hat{r}}{\tilde{C}}$, $\hat{r}\in Locked(\tilde{C})$ $\hat{r}\notin \rho=(O,L,D)$ and $$\tilde{C}, (s,h,(O\cup\{\hat{r}\},L,D))\rightarrow_{p} \abok.$$ 

We have that $r'\notin Res(\tilde{C})\subset Res(C)$. From the induction hypothesis, we conclude that $$\tilde{C}[r'/r], (s,h,(O[r'/r]\cup\{\hat{r}[r'/r]\},L[r'/r],D[r'/r]))\rightarrow_{p} \abok.$$

Note that $C[r'/r]=\resk{\hat{r}[r'/r]}{\tilde{C}[r'/r]}$ and  $$\hat{r}\in Locked(\tilde{C})\ \ \textrm{iff}\ \ \hat{r}[r'/r]\in Locked(\tilde{C}[r'/r]).$$

Therefore $$C[r'/r], (s,h,\rho[r'/r])\rightarrow_{p} \abok.$$

The case $(RA2)$ is analogous to the previous case.

Suppose that the transition is $(WA1)$. We have $C=\withink{\hat{r}}{\tilde{C}}$ and $$\tilde{C}, (s,h,\rho\setminus \{\hat{r}\})\rightarrow_{p} \abok.$$ 

We know that $r'\notin Res(\tilde{C})$ and $r'\notin \rho \setminus \{\hat{r}\}$, because $r'\notin Res(C)$ and $r'\notin \rho$ respectively. 
By induction hypothesis, we have the following transition $$\tilde{C}[r'/r], (s,h,(\rho\setminus \{\hat{r}\})[r'/r])\rightarrow_{p} \abok.$$ 

Note that $(\rho\setminus \{\hat{r}\})[r'/r]=\rho[r'/r]\setminus \{\hat{r}[r'/r]\}$.
Therefore $$C[r'/r], (s,h,\rho[r'/r])\rightarrow_{p} \abok.$$

This concludes the proof of the first equivalence.
For the second equivalence, we suppose that $C, (s,h,\rho)\rightarrow_{p}C', (s',h',\rho')$. And show that $$C[r'/r], (s,h,\rho[r'/r])\rightarrow_{p} C'[r'/r], (s',h',\rho'[r'/r]).$$

Suppose that the transition is given by $(BCT)$, $(LP)$, $(IF1)$, $(IF2)$, $(S1)$ or $(P3)$. The transition does not depend in the resource context or the resource names. So, $C[r'/r], (s,h,\rho[r'/r])\rightarrow_{p}C'[r'/r], (s',h',\rho'[r'/r])$.

Suppose that the transition is given by $(P1)$, $(P2)$ or $(S2)$. Using the induction hypothesis, it is straightforward that $$C[r'/r], (s,h,\rho[r'/r])\rightarrow_{p}C'[r'/r], (s',h',\rho'[r'/r]).$$

Suppose that the transition is given by $(R0)$.
Then $C=\resk{\hat{r}}{\skipk}$, $C'=\skipk$ and $\hat{r}\notin \rho$. Note that $\hat{r}[r'/r]\notin \rho[r'/r]$. Hence $$C[r'/r], (s,h,\rho[r'/r])\rightarrow_{p}C'[r'/r], (s',h',\rho'[r'/r]).$$

Suppose that the transition is given by $(R1)$. Then $C=\resk{\hat{r}}{\tilde{C}}$, $C'=\resk{\hat{r}}{\tilde{C}'}$, $\hat{r}\notin \rho$, $\hat{r}\in Locked(C)$ and $$\tilde{C}, (s,h,(O\cup\{\hat{r}\}, L ,D))\rightarrow_{p} \tilde{C}', (s',h',\rho''),$$ where $\rho''\setminus \{\hat{r}\}=\rho'$. 

Note that $r'\notin Res(\tilde{C})$ and $r'\notin (O\cup\{\hat{r}\}, L ,D)$.
By the induction hypothesis, we have the following transition $$\tilde{C}[r'/r], (s,h,(O\cup\{\hat{r}\}, L ,D)[r'/r])\rightarrow_{p}\tilde{C}'[r'/r], (s',h',\rho''[r'/r]).$$ 
Moreover, we know that $(\rho''\setminus \{\hat{r}\})[r'/r]=\rho''[r'/r]\setminus \{\hat{r}[r'/r]\}=\rho'[r'/r]$ and $\hat{r}[r'/r]\in Locked(C[r'/r])$.
Therefore, $$C[r'/r], (s,h,\rho[r'/r])\rightarrow_{p}C'[r'/r], (s',h',\rho'[r'/r]).$$

The cases $(R2)$ and $(W1)$ are analogous to the previous case.

Suppose that it is $(W0)$. Then  $s(B)=\texttt{true}$, $C=\withk{\hat{r}}{B}{\tilde{C}}$, $C'=\withink{\hat{r}}{\tilde{C}}$, $\rho=(O,L,D\cup \{\hat{r}\})$ and $\rho'=(O\cup \{\hat{r}\},L,D)$.

When $\hat{r}\neq r$, the conclusion is immediate. Suppose that $\hat{r}= r$. 

We have $C[r'/r]=\withk{r'}{B}{\tilde{C}[r'/r]}$, $C'[r'/r]=\withink{r'}{\tilde{C}[r'/r]}$, $\rho[r'/r]=(O,L,D\cup \{r'\})$ and $\rho'[r'/r]=(O\cup \{r'\},L,D)$.

It follows that $$C[r'/r], (s,h,\rho[r'/r])\rightarrow_{p}C'[r'/r], (s',h',\rho'[r'/r]).$$

The case $(W2)$ is analogous to the previous case.
%
%
%
%
\end{proof}

\section{Validity}


In this section, we start by defining the validity of specifications in the SOS presented before. This captures the idea that a specification is valid if and only if every execution starting from a state that respects the precondition and the shared state is not faulty and if it terminates, then the postcondition and the shared state are respected. 

Let $\Gamma$ be a resource context and $D=\{r_{i_{1}}, \ldots, r_{i_{k}}\}\subseteq Res(\Gamma)$, we define 
$$\displaystyle\underset{r\in D}{\circledast} \Gamma(r):=R_{i_{1}}\ast\ldots\ast R_{i_{k}},\quad inv(\Gamma):=\displaystyle\underset{r\in Res(\Gamma)}{\circledast} \Gamma(r).$$

\begin{defi}\label{validity}
We write $\Gamma \models \{P\} C \{Q\}$, if for every  $(s,h)\in\mathcal{S}$ such that $s,h \models P \ast inv(\Gamma)$, we have that
\begin{itemize}
	\item $C,(s,h, (\emptyset,\emptyset,res(\Gamma))\not\rightarrow^{k}_{p} \abok$, for every $k\geq 0$. And
	\item If there exist  $(s',h')\in\mathcal{S}$ and $k\geq 0$ such that $$C,(s,h,(\emptyset,\emptyset,res(\Gamma)))\rightarrow^{k}_{p} \skipk, (s',h',(\emptyset,\emptyset,res(\Gamma))),$$ then $s',h' \models Q \ast inv(\Gamma)$.
\end{itemize}
\end{defi}

However we were not able to inductively prove the soundness of CSL using this notion, because we can not emulate the change of parallel execution in all its parts. The rest of this section is devoted to see how we overcome this difficulty. Thus we introduce the environment transition, that will be essential to spread changes made in the state by one program to other parallel programs. And we give a refined notion of validity for the SOS extended with the environment transition, this new notion is called safety. We finish this section by seeing that safety implies validity.

\subsection{Environment transition}\label{envtrasec}

In order to define the environment transition, we define the environment transformation respecting a set of variables. This transformation modifies  the storage and the resource configuration, afterwards the environment transition combines this transformations with modification in the shared state.

\begin{defi}
Let $(s,h,(O,L,D)), (s',h',(O',L',D'))$ be states and $A\subseteq \textbf{Var}$.
We say that the \emph{environment transforms} $(s, h,(O,L,D))$ into $(s', h,(O,L',D'))$ respecting $A$ and we write $(s, h,(O,L,D)) \stackrel{A}{\leftrightsquigarrow} (s', h,(O,L',D'))$ if $s(x)=s'(x)$, for every $x\in A$, and $L'\cup D' = L\cup D$.
\end{defi}

Note that the environment transformation preserves the local heap and the owned resources, since other programs cannot change them. Furthermore, the environment transformation, $\stackrel{A}{\leftrightsquigarrow}$, naturally defines a relation between states. It is easy to see that this relation is an equivalence relation and it is order reversing with respect to $A$. In the next proposition, we state this properties.

\begin{prop}\label{amb}
Let $A',A\subseteq\textbf{Var}$.
The relation $\stackrel{A}{\leftrightsquigarrow}$ is an equivalence relation. 
If $A'\subseteq A$ and $(s,h,\rho)\stackrel{A}{\leftrightsquigarrow}(s',h',\rho')$, then $(s,h,\rho) \stackrel{A'}{\leftrightsquigarrow}(s',h',\rho')$.
\end{prop}

We denote the \emph{environment transition} by $\xrightarrow{A, \Gamma}_{e}$, it is a relation between $(C,(s,h\uplus h_G,\rho))$ and $(C,(s', h\uplus h'_G,\rho'))$, where $C$ is a command and $(s,h\uplus h_G,\rho)$, $(s', h\uplus h'_G,\rho')$ are states, and it is defined by the rule below.
 Consider the set $A'=A\cup \bigcup_{r\in Locked(C)}PV(r)$. If $(s, h,\rho) \stackrel{A'}{\leftrightsquigarrow} (s', h,\rho')$, $s, h_G\models \underset{r\in D}{\circledast} \Gamma(r)$ and $s', h'_G\models \underset{r\in D'}{\circledast} \Gamma(r)$, then 
$$\infer[(E)]{C, (s,h\uplus h_G,\rho)\xrightarrow{A, \Gamma}_{e} C, (s', h\uplus h'_G,\rho')}{}.$$

As noted before the environment transition is used to simulate modification done by parallel program. The environment transition can be used to: change the storage, except for variables in the rely-set $A$ or variables protected by a locked resources; interchange locked resources and available resource; and modify the available shared heap.

We extend the transitions on the SOS with the environment transition, and we define the relation $\xrightarrow{A, \Gamma}$ from $(C,(s,h,\rho))$ to $(C',(s',h',\rho'))$ or to $\abok$, where $C,C'$ are commands and $(s,h,\rho), (s',h',\rho')$ are states. 
This relation is given by $$\xrightarrow{A, \Gamma}\  =\  \rightarrow_{p}\cup \xrightarrow{A, \Gamma}_{e}.$$

\subsection{Safety}

For a command $C$, we associate the set of variables passive to be \emph{changed} by $C$ in the next transition, and we denote it by $chng(C)$. This set consists of all variables $\textsf{x}$ such that $C$ can perform a transition using $\assigk{x}{e}$; $\loadk{x}{e}$ or $\consk{x}{e}$.

In the next definition of a program's safety with respect to a state for the following $n$ transitions, we include some additional properties that will be useful to prove the soundness of CSL.
\begin{defi}
Let $C\in \mathcal{C}$, $(s,h)\in\mathcal{S}$, $\rho\in\mathcal{O}$, $\Gamma$ be a resource context, $Q$ be an assertion and $A\subseteq \textbf{Var}$. We say that $\mathit{Safe}_0(C, s, h,\rho, \Gamma, Q, A)$ is always valid, and $\mathit{Safe}_{n+1}(C, s, h,\rho, \Gamma, Q, A)$ is valid if:
	\begin{enumerate}
		\item[(i)] If $C=\skipk$, then $s,h \models Q$;
		\item[(ii)]  $C,(s,h,\rho)\not\rightarrow_{p} \abok$;
		\item[(iii)] $chng(C)\cap \bigcup_{r\in L\cup D} PV(r)= \emptyset$;
		\item[(iv)] For every $h_G$ such that $h\bot h_G$, $s,h_G \models  \underset{r\in D}{\circledast} \Gamma(r)$ and $$ C,(s,h\uplus h_G,\rho)\xrightarrow{A, \Gamma} C', (s', \hat{h},\rho'),$$ then there exist $h'$ and $h'_G$ such that $\hat{h}= h'\uplus h'_G$, $$s',h'_G \models \underset{r\in D'}{\circledast} \Gamma(r),\quad \mathit{Safe}_n(C', s', h',\rho', \Gamma, Q, A)\textrm{ is valid.}$$
\end{enumerate}
\end{defi}

The property $(i)$ states that if the execution terminates, then $Q$ is respected. In the property $(ii)$, we ensure that the next transition of $C$ does not abort for the state $(s,h,\rho)$. The property $(iii)$ guarantees that the next transition of $C$ does not change variables protected by resources not owned. In the final condition $(iv)$, we require that the available shared state is preserved after every transition and that the posterior transitions respects this conditions.  

In the next theorem, we see that if a program is safe for every number of transitions and for every state that respects the pre-condition, then the corresponding specification is valid with respect to the SOS. The theorem is proved by induction on the number of program's transitions.

\begin{teo}\label{safetosemop}
Let $C\in\mathcal{C}$, let $P,Q$ be assertions, let $\Gamma$ be a resource context and $A\subseteq \textbf{Var}$. If for every $(s,h)\in\mathcal{S}$ and $n\geq 0$ such that $s,h\models P$, we have that $\mathit{Safe}_n(C,s,h,(\emptyset,\emptyset, Res(\Gamma)), \Gamma, Q, A)$ is valid, then $$\Gamma \models \{P\} C \{Q\}.$$
\end{teo}

In order to prove the soundness of CSL, by the theorem above, is sufficient to show that every derivable specification on CSL implies safety, a result we prove in the next section.

\section{Soundness}\label{soundness}

%

We sketch here the soundness of CSL with respect to the SOS. First, we state the main result of this work, the soundness of CSL. Next we present an intermediate theorem that, together with the Theorem~\ref{safetosemop}, proves the main result. The intermediate theorem says that every derivable specification in CSL is safe in the extended operational semantics.

\begin{teo}\label{main}
If $\Gamma\vdash_{A} \{P\} C \{Q\}$, then $\Gamma\models \{P\} C\{Q\}$.
\end{teo}

Theorem~\ref{main} is an immediate consequence of the next theorem and Theorem~\ref{safetosemop}.

\begin{teo}\label{CSLtosafe}
Let $C$ be a command, let $P,Q$ be assertions, let $\Gamma$ be a resource context and $A\subseteq \textbf{Var}$. If $\Gamma\vdash_{A} \{P\} C \{Q\}$, then  for every $(s,h)\in\mathcal{S}$ and $n\geq 0$ such that $s,h\models P$, we have that $\mathit{Safe}_n(C,s,h,(\emptyset,L,D), \Gamma, Q ,A)$ is valid, where $L\cup D = Res(\Gamma)$.
\end{teo}

In the next lines, we present a proof of this theorem by studying the inference rules of CSL. The proof is carried by induction on the inference rules and uses auxiliary results about the inference rules.


\begin{prop}\label{safeskip}
Let $(s,h)\in \mathcal{S}$, let $\rho\in\mathcal{O}$, let $\Gamma$ be a resource context, let $Q$ be an assertion and $A\subseteq \textbf{Var}$ such that $FV(Q)\subseteq A$. 
If $s,h\models Q$, then $\mathit{Safe}_{n}(\skipk,s,h,\rho, \Gamma, Q, A)$ is valid for every $n\geq 0$.
\end{prop}

\begin{proof}
Let $(s,h)\in \mathcal{S}$, let $\rho\in\mathcal{O}$, let $\Gamma$ be a resource context, let $Q$ be an assertion and $A\subseteq \textbf{Var}$ such that $FV(Q)\subseteq A$ and $s,h\models Q$.

We prove the result by induction on $n$.
For $n=0$, the result is trivial.

Let $n=k+1$.  The first properties of safety are immediate, because $s,h\models Q$, the command $\skipk$ does not abort and it does not modify protected variables.

For the last property, let $h_G$, $C'$, $s'$, $h'$ and $\rho'$ be such that $s,h_G\models \underset{r\in D}{\circledast} \Gamma(r)$ and $$ \skipk,(s,h\uplus h_G,\rho)\xrightarrow{A, \Gamma} C', (s', h',\rho').$$

We note that the only possible transition of $\skipk$ is an environment transition. Then $C'=\skipk$,  $s(x)=s'(x)$, for every $x\in A$, and there is $h'_G$ such that  $h'=h\uplus h'_G$ and $$s',h'_G\models \underset{r\in D'}{\circledast} \Gamma(r).$$ 

It is enough to check that $Safe_k(\skipk, s', h, \rho', \Gamma,Q,A)$ is valid. But the environment transition does not modify the variables in the rely-set neither the local heap. 
Therefore $s',h\models Q$, by Proposition \ref{astneval}. And $Safe_k(\skipk, s', h, \rho', \Gamma,Q,A)$ is valid, by induction.
\end{proof}

In order to check the safety of basic commands rules $(BC)$, we argue mostly as in the context of SL. Like in SL, we know that if a state respects the precondition, then the execution does not abort and the state reached after the program transition $(BCT)$ respects the post condition.

\begin{prop}
Let $c$ be a basic command, let $P,Q$ be assertions, let $\Gamma$ be a resource context, let $(s,h)\in \mathcal{S}$, let $\rho\in\mathcal{O}$ and $A\subseteq \textbf{Var}$ such that $\vdash^{SL} \{P\} c \{Q\}$ and $FV(P,c,Q)\subset A$.
If  $s,h\models P$ and $mod(c)\notin PV(\Gamma)$, then $\mathit{Safe}_n(c,s,h,\rho, \Gamma, Q, A)$ is valid, for all $n\geq 0$.
\end{prop}

\begin{proof}
Let $c$ be a basic command, let $P,Q$ be assertions, let $\Gamma$ be a resource context, let $(s,h)\in \mathcal{S}$, let $\rho\in\mathcal{O}$ and $A\subseteq \textbf{Var}$ such that $\vdash^{SL} \{P\} c \{Q\}$, $FV(P,c,Q)\subset A$, $s,h\models P$ and $mod(c)\notin PV(\Gamma)$.

From $\vdash^{SL} \{P\} c \{Q\}$, we know that $[c](s,h)\neq\abok$ and $[c](s,h)\models Q$.

We prove by induction on $n$ that $\mathit{Safe}_n(c,s,h,\rho, \Gamma, Q, A)$ is valid.
For $n=0$, the result is trivial.

Let $n=k+1$.  The first properties of safety are immediate, because $c\neq \skipk$, $[c](s,h)\neq\abok$ and $mod(c)\notin PV(\Gamma)$.

For the last property, let $h_G$, $C'$, $s'$, $h'$ and $\rho'$ be such that $s,h_G\models \underset{r\in D}{\circledast} \Gamma(r)$ and $$ c,(s,h\uplus h_G,\rho)\xrightarrow{A, \Gamma} C', (s', h',\rho').$$

We have two possibilities: a transition by the environment or by $(BCT)$.

Suppose that it is an environment transition, then $C'=c$,  $s(x)=s'(x)$, for every $x\in A\supseteq FV(P)$, and there is $h'_G$ such that  $h'=h\uplus h'_G$ and $$s',h'_G\models \underset{r\in D'}{\circledast} \Gamma(r).$$
Moreover, the precondition is preserved, $s',h\models P$. So we can apply the induction hypothesis to see that 
$\mathit{Safe}_n(c,s',h,\rho', \Gamma, Q, A)$ is valid. 

Suppose that it is a transition by $(BCT)$, then $C'=\skipk$, $(s',h')=[c](s,h\uplus h_G)$ and $\rho'=\rho$. Using the frame property, Proposition~\ref{safetymono}, we know that $h'=h''\uplus h_G$, where $[c](s,h)=(s',h'')\models Q$. 

Therefore, by Proposition~\ref{safeskip}, we have that  $\mathit{Safe}_n(\skipk,s',h'',\rho, \Gamma, Q, A)$ is valid.
This concludes the proof.
\end{proof}

The soundness of the rules $(SEQ)$, $(LP)$, $(CONJ)$, $(IF)$ and $(CONS)$ are a direct consequence of the inductive process.

The soundness of the frame rule $(FRA)$ is supported by the following proposition. It follows from the safety monotonicity and frame property (Propositions \ref{safetymono} and \ref{frameprop}). We note that $R$ is valid after every transition, because $FV(R)$ is not modified by the command and the rely-set includes it.

\begin{prop}\label{safeframe}
Let $C$ be a reachable command, let $\Gamma$ be a resource context, let $(s,h\uplus h_R)\in\mathcal{S}$, let $\rho\in\mathcal{O}$, let $Q,R$ be assertions and $A\subseteq \textbf{Var}$ such that $s,h_R\models R$.
If  $\mathit{Safe}_n(C,s,h,\rho, \Gamma, Q,A)$ is valid and $mod(C)\cap FV(R)=\emptyset$, then $\mathit{Safe}_n(C,s,h\uplus h_R,\rho, \Gamma, Q\ast R, A\cup FV(R))$ is valid.
\end{prop}

\begin{proof}
Let $C$ a reachable command, $\Gamma$ a resource context, $(s,h\uplus h_R)\in\mathcal{S}$,  $\rho\in\mathcal{O}$, $Q,R$ assertions and $A\subseteq \textbf{Var}$ such that $s,h_R\models R$,  $\mathit{Safe}_n(C,s,h,\rho, \Gamma, Q,A)$ is valid and $mod(C)\cap FV(R)=\emptyset$.

We just show the inductive step of the proof.
Let $n=k+1$.

If $C=\skipk$, then $s,h\models Q$ and $s,h\uplus h_R\models Q\ast R$.

By $Safe_n(C,s,h,\rho, \Gamma, Q,A)$ and Proposition \ref{safetymono}, we have that $$C,(s,h\uplus h_R,\rho)\not\rightarrow \abok,$$
and $$chng(C)\cap \bigcup_{r\in L\cup D} PV(r)= \emptyset.$$
Therefore, the first properties of $Safe_n(C,s,h\uplus h_R,\rho, \Gamma, Q\ast R,A\cup FV(R))$ are established.

In order to check the last property. Let $h_G$, $C'$, $s'$, $h'$ and $\rho'$ such that $h_G\bot (h\uplus h_R)$, $s,h_G\models \underset{r\in D}{\circledast} \Gamma(r)$ and  
$$C,(s,h\uplus h_R\uplus h_G,\rho)\xrightarrow{A\cup FV(R), \Gamma} C', (s', h',\rho').$$

This transition can be a program transition or a environment transition.

Suppose that it is a program transition.

By frame property, Proposition~\ref{safetymono}, there is $h''$ such that $h'=h''\uplus h_R$ and $$C,(s,h\uplus h_G,\rho)\rightarrow_p C', (s', h'',\rho').$$ 

From the validity of $Safe_n(C,s,h,\rho, \Gamma, Q,A)$ and the transition above, we know that there exists $h'_G\subseteq h''$, such that $Safe_k(C',s',h''\setminus h'_G,\rho', \Gamma, Q,A)$ is valid and $$s',h'_G\models \underset{r\in D'}{\circledast} \Gamma(r).$$

From $mod(C)\cap FV(R)=\emptyset$, we have that $mod(C')\cap FV(R)=\emptyset$ and $$s',h_R\models R.$$ 

By induction, $Safe_k(C',s',h'\setminus h'_G,\rho', \Gamma, Q\ast R,A\cup FV(R))$ is valid.

Suppose that it occurs an environment transition.

There exists $h'_G\subseteq h'$ such that $h'=h\uplus h_R\uplus h'_G$, $s',h'_G\models \underset{r\in D'}{\circledast} \Gamma(r)$,
and $$(s,h\uplus h_R,\rho)\stackrel{A'}{\leftrightsquigarrow} (s',h\uplus h_R,\rho'),$$
where $A'=A\cup FV(R)\cup \bigcup_{r\in Locked(C)}PV(r)$. And $s',h_R\models R$, because $FV(R)\subseteq A'$.

Let $A''= A\cup \bigcup_{r\in Locked(C)}PV(r)\subset A'$. Then $(s,h,\rho)\stackrel{A''}{\leftrightsquigarrow} (s',h,\rho')$, and 
$$C,(s,h\uplus h_G,\rho)\xrightarrow{A, \Gamma}_{e} C, (s', h\uplus h_G,\rho').$$

From the validity of $Safe_n(C,s,h,\rho, \Gamma, Q,A)$ and the transition above, we have that $Safe_k(C,s',h,\rho', \Gamma, Q,A)$ is valid.

By induction, we conclude that $Safe_k(C,s',h\uplus h_R,\rho', \Gamma, Q\ast R,A\cup FV(R))$ is valid.
\end{proof}

Next we study the parallel rule $(PAR)$.

\begin{prop}\label{safepar}
Let $\park{C_1}{C_2}$ be a reachable command, let $Q_1,Q_2$ be assertions, let $(s,h),(s,h_1),(s,h_2)\in\mathcal{S}$, let $\rho,\rho_1,\rho_2\in\mathcal{O}$ and $A_1,A_2\subseteq \textbf{Var}$. Suppose that $h=h_1\uplus h_2$, $\rho=(O_1\cup O_2,L,D)$, $\rho_1=(O_1,L\cup O_2,D)$, $\rho_2=(O_2,L\cup O_1,D)$, $FV(Q_i)\subseteq A_i$, for $i=1,2$, and $A_1 \cap mod(C_2)=A_2 \cap mod(C_1)=\emptyset$.

If $\mathit{Safe}_n(C_1, s, h_1, \rho_1, \Gamma, Q_1, A_1)$ and $\mathit{Safe}_n(C_2, s, h_2, \rho_2, \Gamma, Q_2, A_2)$ are valid , then $\mathit{Safe}_n(\park{C_1}{C_2}, s, h,\rho, \Gamma, Q_1 \ast Q_2, A_1\cup A_2)$ is valid.
\end{prop}

As before we prove this result by induction on $n$. The firsts three properties of safety are immediate from the safety of $C_1$ and $C_2$, and Propositions \ref{safetymono} and \ref{abortpar}.

In order to apply the induction step we use the environment transition. If the parallel execution transits by a program transition, then there are three cases. First case, a transition is done by $C_1$. We perform the same transition on $C_1$ (by Propositions \ref{frameprop} and \ref{p1}) and an environment transition on $C_2$, that replicates the changes performed by the program transition. This environment transition exists because the variables modified by the program $C_1$ are different from the rely-set $A_2$. In the second case, a transition is done by $C_2$, and we do analogous transitions. The third case is the joint of parallel commands. In this case, we do a reflexive environment transition on $C_1$ and $C_2$. If the program transits by an environment transition, then we perform the same environment transition on $C_1$ and $C_2$. This environment transition can be used because the rely-set of $\park{C_1}{C_2}$ includes the rely-set of each command. Therefore we can apply the inductive hypothesis and obtain the proposition.

\begin{proof}
Let $\park{C_1}{C_2}$ be a reachable command, let $(s,h)(s,h_1),(s,h_2)\in\mathcal{S}$, let $\rho,\rho_1,\rho_2\in\mathcal{O}$, let $Q_1,Q_2$ be assertions and $A_1,A_2\subseteq \textbf{Var}$ such that $h=h_1\uplus h_2$, $\rho=(O_1\cup O_2,L,D)$, $\rho_1=(O_1,L\cup O_2,D)$, $\rho_2=(O_2,L\cup O_1,D)$, $FV(Q_1)\subseteq A_1$, $FV(Q_2)\subseteq A_2$ and $A_1 \cap mod(C_2)=A_2 \cap mod(C_1)=\emptyset$.

We just prove the induction step for $n=k+1$. 

Suppose that the hypothesis is valid, i.e. $Safe_{k+1}(C_1,s,h_1,\rho_1, \Gamma, Q_1, A_1)$ and $Safe_{k+1}(C_2,s,h_2,\rho_2, \Gamma, Q_2, A_2)$ are valid.

The first property of $Safe_n(\park{C_1}{C_2}, s, h,\rho, \Gamma, Q_1 \ast Q_2, A_1\cup A_2)$ is trivial, because $\park{C_1}{C_2}\neq\skipk$.

Applying the safety monotonicity (Proposition $\ref{safetymono}$) to the hypothesis, we see that
$$C_1, (s, h_1\uplus h_2,\rho_1) \not\rightarrow_{p} \abok \quad C_2, (s, h_1\uplus h_2,\rho_2) \not\rightarrow_{p} \abok.$$

Using Proposition $\ref{abortpar}$, we obtain that $$C_1, (s, h,\rho) \not\rightarrow_{p} \abok \quad C_2, (s, h,\rho) \not\rightarrow \abok.$$ 
Hence, we respect the second property $\park{C_1}{C_2}, (s, h, \rho) \not\rightarrow_{p} \abok$.

Using the hypothesis, we derive that
\begin{align*}
chng(C)\cap \bigcup_{\substack{r\in \\L \cup D}} PV(r)= & (chng(C_1)\cup chng(C_2))\cap \bigcup_{\substack{r\in \\L \cup D}} PV(r)\\
=& (chng(C_1) \cap \bigcup_{\substack{r\in \\ L \cup D}} PV(r))\cup (chng(C_2) \cap \bigcup_{\substack{ r\in \\ L \cup D}} PV(r))\\
\subseteq &  (chng(C_1) \cap \bigcup_{\substack{ r\in L\cup \\  O_2 \cup D}} PV(r))\cup (chng(C_2) \cap \bigcup_{\substack{ r\in L\cup\\  O_1 \cup D}} PV(r))\\
\subseteq & \emptyset.
\end{align*}
And the third condition of $Safe_n(\park{C_1}{C_2}, s, h,\rho, \Gamma, Q_1 \ast Q_2, A_1\cup A_2)$ follows.

Now, we check the fourth condition. Let $h_G$, $C'$, $s'$, $h'$ and $\rho'$ such that $s,h_G\models \underset{r\in D}{\circledast} \Gamma(r)$ and $$ \park{C_1}{C_2},(s,h_1\uplus h_2\uplus h_G,\rho)\xrightarrow{A_1\cup A_2, \Gamma} C', (s', h',\rho').$$

There are four possible transitions.

Suppose that the transition is given by $(P1)$. We have that $C'=\park{C'_1}{C_2}$ and $$C_1,(s,h_1\uplus h_2\uplus h_G,\rho)\rightarrow_{p} C'_1, (s', h', \rho').$$


%

The validity of $Safe_{k+1}(C_1,s,h_1,\rho_1, \Gamma, Q_1, A_1)$ implies that the command $C_1$ does not abort. Using the frame property (Proposition $\ref{frameprop}$) and Proposition $\ref{p1}$ we have that $O_2\subseteq O'$, $h_2\subseteq h'$ and $$C_1,(s,h_1\uplus h_G,\rho_1)\rightarrow_{p} C'_1, (s', h'\setminus h_2, (O'\setminus O_2,L\cup O_2,D')).$$

Define $\rho'_1:=(O'\setminus O_2,L\cup O_2,D')$. From the hypothesis, we know that there are $h'_1$ and $h'_G$ such that $h'\setminus h_2=h'_1\uplus h'_G$, $Safe_{k}(C'_1,s',h'_1, \rho'_1,\Gamma, Q_1,A_1)$ is valid and $$s',h'_G\models \underset{r\in D'}{\circledast} \Gamma(r).$$

In order to apply the induction hypothesis and conclude the validity of $Safe_k(\park{C'_1}{C_2}, s', h'_1\uplus h_2,\rho', \Gamma, Q_1\ast Q_2, A_1\cup A_2)$, it remains to check that $Safe_k(C_2, s',  h_2,\rho_2', \Gamma,  Q_2,  A_2)$ is valid, where $\rho'_2=(O_2,L\cup (O'\setminus O_2),D')$.

We know that $s(x)=s'(x)$, for every $x\notin chng(C_1)$.

Note that $chng(C_1)\subset mod(C_1)$ and $Locked(C_2)\subset O_2$, otherwise $C_2$ will abort by $(WA2)$. Let $A'_2=A_2\cup \bigcup_{r\in Locked(C_2)}PV(r)$, using the hypothesis we have that 

$$chng(C_1)\cap A'_2\subseteq (mod(C_1)\cap A_2) \cup(chng(C_1) \cap \bigcup_{r\in O_2}PV(r))\subseteq \emptyset.$$

Hence $s(x)=s'(x)$, for every $x\in A'_2$.

From Proposition $\ref{proprec}$, we know that $O_1 \cup D=(O'\setminus O_2) \cup D'$ and $$L\cup O_1 \cup D=L\cup (O'\setminus O_2) \cup D'.$$

We have the environment transformation and environment transition  $$(s, h_2, \rho_2)\stackrel{A'_2}{\leftrightsquigarrow} (s', h_2, \rho'_2),$$ 
$$C_2,(s, h_2\uplus h_G, \rho_2)\xrightarrow{A_2, \Gamma}_{e}C_2,(s', h_2\uplus h'_G, \rho'_2),$$

By hypothesis we conclude that $Safe_{k}(C_2,s',h_2, \rho'_2, \Gamma, Q_2, A_2)$ is valid. 

Therefore $Safe_k(\park{C'_1}{C_2}, s', h'_1\uplus h_2,\rho', \Gamma, Q_1\ast Q_2, A_1\cup A_2)$ is valid.

The transition $(P2)$ is analogous to the transition $(P1)$.

Suppose that the transition is given by $(P3)$. We have $C'=\skipk$, $C_1=\skipk$, $C_2=\skipk$, $s'=s$, $h'=h\uplus h_G$ and $\rho'=\rho$. 

Taking $h'_G=h_G$. We know that $$s,h_G\models \underset{r\in D}{\circledast} \Gamma(r).$$

Because $Safe_{n}(\skipk,s, h_1,\rho_1, \Gamma, Q_1, A_1)$ and $Safe_{n}(\skipk,s, h_2,\rho_2 \Gamma, Q_2, A_2)$ are valid, we have that 
$$s,h_1\uplus h_2 \models Q_1 \ast Q_2.$$

From  Proposition $\ref{safeskip}$, $Safe_k(\skipk, s, h_1\uplus h_2,\rho, \Gamma, Q_1\ast Q_2, A_1\cup A_2)$ is valid.

Suppose that the transition is given by $(E)$. Let $$A'=A_1\cup A_2 \cup \bigcup_{r\in Locked(\park{C_1}{C_2})} PV(r).$$

 We have that $C'=\park{C_1}{C_2}$, $\rho'=(O_1\cup O_2,L',D')$, $(s,h,\rho)\stackrel{A'}{\leftrightsquigarrow} (s',h,\rho')$, and there exists $h'_G$ such that $h'=h_1\uplus h_2 \uplus h'_G$ and $$s', h'_G \models \underset{r\in D'}{\circledast} \Gamma(r).$$

We know that $s(x)=s'(x)$, for every $x\in A'$ and $L\cup D=L'\cup O_2\cup D'$. Then we have the following environment transformation $$(s,h_1,\rho_1)\stackrel{A'_1}{\leftrightsquigarrow} (s',h_1,\rho'_1),$$ where $\rho'_1=(O_1,L'\cup O_2,D')$ and $A'_1=A_1\cup \bigcup_{r\in Locked(C_1)} PV(r)\subseteq A'$.


By the hypothesis, we conclude that $Safe_k(C_1,s',h_1,\rho'_1,\Gamma,Q_1,A_1)$ is valid.

Analogous, we obtain that $Safe_k(C_2,s',h_2,\rho'_2,\Gamma,Q_2,A_2)$ is valid, where $\rho'_2=(O_2,L'\cup O_1,D')$.


Therefore $Safe_k(\park{C_1}{C_2},s',h_1\uplus h_2, \rho',\Gamma,Q_1\ast Q_2, A_1\cup A_2)$ is valid, by the induction hypothesis.
\end{proof}

The safety of the critical region rule follows from the safety inside the critical region. Because any environment transition performed before the critical region does not break the precondition's validity and when a program enters a critical region its invariant is valid. Therefore the next result establishes safety for the critical region rule.

\begin{prop}\label{safewith}
Let $C$ be a reachable command, let $(s,h)\in \mathcal{S}$, let $\Gamma$ be a resource context, let $\rho=(O,L,D)\in \mathcal{O}$, let $Q, R$ be assertions and $A\subseteq\textbf{Var}$. Suppose that $r\in O$, $\Gamma'=\Gamma,r(X):R$ is a resource context and $FV(Q)\subseteq A$.

If $\mathit{Safe}_n(C,s,h,\rho\setminus \{r\},\Gamma, Q\ast R,A\cup X)$ is valid, then $$\mathit{Safe}_n(\withink{r}{C},s,h,\rho,\Gamma',Q,A)\textrm{ is valid.}$$
\end{prop}

\begin{proof}
Let $C$ be a reachable command, let $(s,h)\in \mathcal{S}$, let $\Gamma$ be a resource context, let $\rho=(O,L,D)\in \mathcal{O}$, let $Q, R$ be assertions and $A\subseteq\textbf{Var}$ such that $r\in O$, $\Gamma'=\Gamma,r(X):R$ is a resource context, $Safe_n(C,s,h,\rho\setminus \{r\},\Gamma, Q\ast R,A\cup X)$ is valid and $FV(Q)\subseteq A$.

We prove by induction on $n$ that $Safe_n(\withink{r}{C},s,h,\rho,\Gamma',Q,A)$ is valid. For $n=0$, it is trivial.

Let $n=k+1$.
The first property is immediate, because $\withink{r}{C} \neq \skipk$.

From $Safe_n(C,s,h,\rho\setminus \{r\},\Gamma, Q\ast R,A\cup X)$ be valid, we know that $$C,(s,h,\rho\setminus \{r\})\not\rightarrow_p \abok.$$

Together with $r\in O$, we have the second property $$\withink{r}{C},(s,h,\rho)\not\rightarrow_p \abok.$$

From $Safe_n(C,s,h,\rho\setminus \{r\},\Gamma, Q\ast R,A\cup X)$ be valid and $r\in O$, we know that $\withink{r}{C}$ does not change variable protected by resource in $L\cup D$.
Then, the third condition is respected.

Let $h_G$, $C'$, $s'$, $h'$ and $\rho'$ such that $s,h_G\models\underset{\hat{r}\in D}{\circledast} \Gamma'(\hat{r})$ and
$$ \withink{r}{C},(s,h\uplus h_G,\rho)\xrightarrow{A, \Gamma'} C', (s', h',\rho').$$

There are three possible transition: $(W1)$, $(W2)$ or $(E)$.

Suppose that it is $(W1)$.
We have $C'=\withink{r}{\tilde{C}}$, $r\in O\cap O'$ and $$C,(s,h\uplus h_G,\rho\setminus \{r\})\rightarrow_p \tilde{C}, (s', h',\rho'\setminus \{r\}).$$  

From $Safe_n(C,s,h,\rho\setminus \{r\},\Gamma, Q\ast R,A\cup X)$, we know that there are $h'_G$, $h'_L$ such that $h'=h'_L\uplus h'_G$, $Safe_{k}(\tilde{C},s',h'_L,\rho'\setminus \{r\},\Gamma, Q\ast R,A\cup X)$ is valid and  $$s',h'_G\models\underset{\hat{r}\in D'}{\circledast} \Gamma(\hat{r}).$$

From $r\in O'$, we conclude that $s',h'_G\models\underset{\hat{r}\in D'}{\circledast} \Gamma'(\hat{r})$.

From $Safe_{k}(\tilde{C},s',h'_L,\rho'\setminus \{r\},\Gamma, Q\ast R,A\cup X)$ and the induction hypothesis, we obtain that $Safe_{k}(\withink{r}{\tilde{C}},s',h'_L,\rho',\Gamma', Q,A)$ is valid.

Suppose that the transition is given by $(W2)$.

We have that $C'=\skipk$, $C=\skipk$, $s=s'$, $h'=h\uplus h_G$, $O'=O\setminus \{r\}$, $L'=L$ and $D'=D\cup \{r\}$. From $Safe_n(\skipk,s,h,\rho,\Gamma, Q\ast R,A\cup X)$, we have that $$s,h\models Q \ast R.$$

Then there exists $h_R\subseteq h$ such that $s,h_R\models R$ and $s,h\setminus h_R\models Q$. 
Then $$s,h_G\uplus h_R\models\underset{\hat{r}\in D'}{\circledast} \Gamma'(\hat{r}).$$

By Proposition $\ref{safeskip}$, $FV(Q)\subseteq A$ and $s,h\setminus h_R\models Q$, we conclude that $Safe_{k}(\skipk, s, h\setminus h_R, \rho',\Gamma', Q, A)$ is valid.

Suppose that the transition is given by $(E)$. Let $$A'=A\cup PV(r) \cup \bigcup_{\hat{r}\in Locked(C)}PV(\hat{r}).$$ 

We have $C'=\withink{r}{C}$, $(s,h,\rho)\stackrel{A'}{\leftrightsquigarrow} (s',h,\rho')$ and there exists $h'_G\subseteq h'$ such that $h'=h\uplus h'_G$ and $$s',h'_G\models \underset{\hat{r}\in D'}{\circledast} \Gamma'(\hat{r}).$$

From $r\in O\cap O'$, we get that $s',h'_G\models \underset{\hat{r}\in D'}{\circledast} \Gamma(\hat{r})$ and $s,h_G\models \underset{\hat{r}\in D}{\circledast} \Gamma(\hat{r})$.

Note that $A'=A\cup X \cup \bigcup_{\hat{r}\in Locked(C)}PV(\hat{r})$ and the environment transition 
$$C,(s,h\uplus h_G, \rho\setminus \{r\})\xrightarrow{A\cup X, \Gamma}_{e} C,(s',h\uplus h'_G, \rho'\setminus \{r\}).$$

From the validity of $Safe_n(C,s,h,\rho\setminus \{r\}, \Gamma, Q\ast R, A\cup X)$, we have that $Safe_{k}(C,s',h,\rho'\setminus \{r\}, \Gamma, Q\ast R, A\cup X)$ is valid. 

Therefore, by induction, $Safe_{k}(\withink{r}{C},s',h,\rho',\Gamma',Q,A)$ is valid.
\end{proof}


In the proposition below, we give properties for the local resource when the resource is available or locked, similar to \cite[Lemma 4.3]{vaf} in DCSL. The soundness of the local resource rule follows from the second property of the proposition.

\begin{prop}\label{saferes}
Let $C$ be a reachable command, let $(s,h)\in\mathcal{S}$, let $\Gamma$ be a resource context, $\rho=(O,L,D)\in\mathcal{O}$,  let $Q,R$ be assertions and $A, X\subseteq\textbf{Var}$. Suppose that $r\notin\rho$, $\Gamma'=\Gamma, r(X):R$ is a resource context and $FV(Q)\subseteq A$. We have the following statements:
\begin{itemize}
	\item Suppose that $r\in Locked(C)$.
	If $\mathit{Safe}_n(C,s,h,(O\cup \{r\}, L,D),\Gamma', Q, A)$ is valid, then $\mathit{Safe}_n(\resk{r}{C},s,h,\rho,\Gamma, Q\ast R, A\cup X)$ is valid.
	\item Suppose that $r\notin Locked(C)$ and that there exists $h_R$ such that $h_R\bot h$ and $s,h_R\models R$.
	If $\mathit{Safe}_n(C,s,h,(O, L,D\cup \{r\}),\Gamma', Q, A)$ is valid, then $\mathit{Safe}_n(\resk{r}{C},s,h\uplus h_R,\rho,\Gamma, Q\ast R, A\cup X)$ is valid.
\end{itemize}
\end{prop}

This proposition is proved by induction on both properties in the following way: first we prove that both properties are true when $n=0$; then we assume that both properties are true for $n\geq 0$ and prove that each property is true for $n+1$. 

The program transitions inside the local resource have an equivalent program transition for the command $C$, except for the transition $(R0)$. In those cases we apply one of the inductive step depending on resource's ownership. For the case $(R0)$, we note that the execution inside the local resource had terminated and the invariant $R$ is respected. If the local resource transits by an environment transition, then there is an equivalent environment transition in $C$.

\begin{proof}
Let $C$ be a reachable command, let $(s,h)\in\mathcal{S}$, let $\Gamma$ be a resource context, $\rho=(O,L,D)\in\mathcal{O}$,  let $Q,R$ be assertions and $A, X\subseteq\textbf{Var}$ such that $r\notin\rho$, $\Gamma'=\Gamma, r(X):R$ is a resource context and $FV(Q)\subseteq A$.

Consider the next statements:

\begin{tabular}{c p{0.82\linewidth}}
$P(n)$ & If $r\in Locked(C)$ and $Safe_n(C,s,h,(O\cup \{r\}, L,D),\Gamma', Q, A)$ is valid, then
$Safe_n(\resk{r}{C},s,h,\rho,\Gamma, Q\ast R, A\cup X)$ is valid.\\
$Q(n)$ & If $r\notin Locked(C)$, $Safe_n(C,s,h,(O, L,D\cup \{r\}),\Gamma', Q, A)$ is valid and there exist $h_R\bot h$ such that $s,h_R\models R$, then $Safe_n(\resk{r}{C},s,h\uplus h_R,\rho,\Gamma, Q\ast R, A\cup X)$ is valid.
\end{tabular}

We prove the result in three steps. First, we note that $P(0)$ and $Q(0)$ are true. Next, we see that $P(n)\wedge Q(n)\Rightarrow Q(n+1)$, for every $n\geq 0$. Last, we show that $P(n)\wedge Q(n)\Rightarrow P(n+1)$, for every $n\geq 0$.

Next, we suppose that $P(n) \wedge Q(n)$ is valid and show that $Q(n+1)$ is valid.

Suppose that $Safe_{n+1}(C,s,h,(O, L,D\cup \{r\}),\Gamma', Q, A)$ is valid, $s,h_R\models R$, $h_R\bot h$  and $r\notin Locked(C)$.

The first property of $Safe_{n+1}(\resk{r}{C},s,h\uplus h_R,\rho,\Gamma, Q\ast R, A\cup X)$ is immediate, because $\resk{r}{C}\neq \skipk$.

From $Safe_{n+1}(C,s,h,(O, L,D\cup \{r\}),\Gamma', Q, A)$ and Proposition \ref{safetymono}, we have $$C,(s,h\uplus h_R,(O,L,D\cup \{r\}))\not\rightarrow_{p} \abok.$$

Note that $r\notin\rho\cup Locked(C)$. Hence $\resk{r}{C}, (s,h\uplus h_R, \rho)\not\rightarrow_{p} \abok$.
And the second property is valid.

From $Safe_{n+1}(C,s,h,(O, L,D\cup \{r\}),\Gamma', Q, A)$, we know that $$chng(\resk{r}{C})\cap \bigcup_{\hat{r}\in L\cup D} PV(\hat{r})\subseteq chng(C)\cap \bigcup_{\hat{r}\in L\cup D\cup \{r\}} PV(\hat{r})=\emptyset.$$
 
Hence, the third property is respected.

 Let $h_G$, $C'$, $s'$, $h'$ and $\rho'$ such that $s,h_G\models\underset{\hat{r}\in D}{\circledast} \Gamma(\hat{r})$ and
$$ \resk{r}{C},(s,h\uplus h_R\uplus h_G,\rho)\xrightarrow{A\cup X, \Gamma} C', (s', h',\rho').$$

Next, we study the possible transitions: $(R0)$, $(R2)$ or $(E)$.

Suppose that the transition is given by $(R0)$.
We have that $C=C'=\skipk$, $s'=s$, $h'=h\uplus h_R\uplus h_G$ and $\rho'=\rho$. 
Consider $h'_G=h_G$. Then $h'_G\subseteq h'$ and $$s,h'_G\models\underset{\hat{r}\in D}{\circledast} \Gamma(\hat{r}).$$

From $Safe_{n+1}(\skipk,s,h,(O, L,D\cup \{r\}),\Gamma', Q, A)$, we have that $s,h\models Q$.
 And $s,h\uplus h_R \models Q\ast R$.

Hence $Safe_n(\skipk, s, h\uplus h_R, \Gamma, Q\ast R, A\cup X)$ is valid, by Proposition $\ref{safeskip}$.

Suppose that it is $(R2)$. We have that $C'=\resk{r}{\tilde{C}}$ and $$C,(s,h\uplus h_R\uplus h_G,(O,L,D\cup \{r\}))\rightarrow_p \tilde{C}, (s', h',\rho''),$$
where $\rho'=\rho''\setminus\{r\}$. 
Note that $s,h_R\uplus h_G\models \underset{\hat{r}\in D\cup \{r\}}{\circledast} \Gamma'(\hat{r})$. 

From $Safe_{n+1}(C,s,h,(O, L,D\cup \{r\}),\Gamma', Q, A)$ and the transition above, we know that there is $h'_G$ such that $h'_G\subseteq h'$, $Safe_n(\tilde{C},s', h'\setminus h'_G,\rho'', \Gamma', Q,A)$ is valid and $$s',h'_G\models \underset{\hat{r}\in D''}{\circledast} \Gamma'(\hat{r}).$$ 

In order to prove that $Safe_n(\resk{r}{\tilde{C}},s',h'\setminus h'_G,\rho',\Gamma, Q\ast R, A\cup X)$ is valid, we need to apply the hypothesis $P(n)$ or $Q(n)$, respectively, if $r\in O''$ or $r\in D''$. Note that $r\notin L''$, by Proposition \ref{proprec}.

We observe that $r\in O''$ if and only if the resource $r$ was acquired in the transition above. 

If $r\in O''$, then $r\in Locked(\tilde{C})$. From $Safe_n(\tilde{C},s', h'\setminus h'_G,\rho'', \Gamma', Q,A)$ and $P(n)$, we have $Safe_n(\resk{r}{\tilde{C}},s',h'\setminus h'_G,\rho',\Gamma, Q\ast R, A\cup X)$ is valid.

If $r\in D''$, then $r\notin Locked(\tilde{C})$. We remark that
$$s',h'_G\models \underset{\hat{r}\in D''}{\circledast} \Gamma'(\hat{r})=R \ast \left(\underset{\hat{r}\in D'}{\circledast} \Gamma(\hat{r})\right).$$ 

Then there exists $h'_R\subseteq h'_G$ such that $s',h'_R\models R$. 
Therefore, by $Q(n)$, $Safe_n(\resk{r}{\tilde{C}},s',h'\setminus h'_G\uplus h'_R,\rho',\Gamma, Q\ast R, A\cup X)$ is valid.


Suppose that the transition is given by $(E)$. Let $$A'=A\cup X \cup \bigcup_{\hat{r}\in Locked(\resk{r}{C})} PV(\hat{r}).$$

We have $C'= \resk{r}{C}$, $(s,h\uplus h_R,\rho)\stackrel{A'}{\leftrightsquigarrow} (s',h\uplus h_R,\rho')$ and there exists $h'_G\subseteq h'$ such that $h'=h\uplus h_R\uplus h'_G$ and $$s',h'_G\models \underset{\hat{r}\in D'}{\circledast} \Gamma(\hat{r}).$$

Let $A''= A\cup \bigcup_{\hat{r}\in Locked(C)} PV(\hat{r})$.
From $r\notin Locked(C)$, we have that $Locked(C)=Locked(\resk{r}{C})$ and $A''\subseteq A'$. Therefore
$$(s,h,(O,L,D\cup\{r\}))\stackrel{A''}{\leftrightsquigarrow} (s',h,(O',L',D'\cup\{r\})).$$

From $FV(R)\subseteq X\subseteq A'$ and Proposition \ref{astneval}, we have that $s',h_R\models R$. 
Moreover, we know that
$$s,h_G\uplus h_R\models  \underset{\hat{r}\in D\cup\{r\}}{\circledast} \Gamma'(\hat{r}), \ \ s',h'_G\uplus h_R\models  \underset{\hat{r}\in D'\cup\{r\} }{\circledast} \Gamma'(\hat{r}).$$

Then, we have the following environment transition $$ C,(s,h\uplus h_R\uplus h_G,(O,L,D\cup\{r\})) \xrightarrow{A, \Gamma'}_e C,(s',h\uplus h_R\uplus h'_G,(O',L',D'\cup\{r\})).$$

From $Safe_{n+1}(C,s,h,(O, L,D\cup \{r\}),\Gamma', Q, A)$ and the environment transition above, it follows that $Safe_{n}(C,s',h,(O', L',D'\cup \{r\}),\Gamma', Q, A)$ is valid.

By $Q(n)$, we have $Safe_{n}(\resk{r}{C},s',h\uplus h_R,\rho',\Gamma, Q\ast R, A\cup X)$.

To finish, we prove that $P(n)\wedge Q(n)$ implies $P(n+1)$.

Suppose that $r\in Locked(C)$ and $Safe_{n+1}(C,s,h,(O\cup \{r\}, L,D),\Gamma', Q, A)$ is valid. Analogous to the previous case, we prove the first three properties of $Safe_{n+1}(\resk{r}{C},s,h,(O, L,D),\Gamma, Q\ast R, A\cup X)$.

Let $h_G$, $C'$, $s'$, $h'$ and $\rho'$ such that $s,h_G\models\underset{\hat{r}\in D}{\circledast} \Gamma(\hat{r})$ and
$$ \resk{r}{C},(s,h\uplus h_G,\rho)\xrightarrow{A\cup X, \Gamma} C', (s', h',\rho').$$

There are two possible transitions: $(R1)$ or $(E)$.

Suppose that the transition is  $(R1)$.
We have $C'=\resk{r}{\tilde{C}}$ and
$$C,(s,h\uplus h_G,(O\cup \{r\},L,D))\rightarrow_p \tilde{C}, (s', h',\rho''),$$
where $\rho'=\rho''\setminus\{r\}$. Note that $\underset{\hat{r}\in D}{\circledast} \Gamma(\hat{r})=\underset{\hat{r}\in D}{\circledast} \Gamma'(\hat{r})$.

From the validity of $Safe_{n+1}(C,s,h,(O\cup \{r\}, L,D),\Gamma', Q, A)$ and the transition above, there is $h'_G$ such that $h'_G\subseteq h'$, $Safe_{n}(\tilde{C},s',h'\setminus h'_G,\rho'',\Gamma', Q, A)$ is valid and $$s',h'_G\models\underset{\hat{r}\in D''}{\circledast} \Gamma'(\hat{r}).$$

As before, we study two cases: $r\in O''$ or $r\in D''$.
In this case, we observe that $r\in D''$ if and only if the resource $r$ was released in the transition above.

If $r\in O''$, then $r\in Locked(\tilde{C})$ and $s',h'_G\models\underset{\hat{r}\in D'}{\circledast} \Gamma(\hat{r})$.

By $P(n)$, we obtain $Safe_{n}(\resk{r}{\tilde{C}},s',h'\setminus h'_G,\rho',\Gamma, Q\ast R, A\cup X)$.

If $r\in D''$, then $r\notin Locked(\tilde{C})$ and $\underset{\hat{r}\in D''}{\circledast} \Gamma'(\hat{r})=R\ast \underset{\hat{r}\in D'}{\circledast} \Gamma(\hat{r})$. Moreover there exists $h'_R\subseteq h'_G$ such that $s',h'_R\models R$.

By $Q(n)$, we have $Safe_{n}(\resk{r}{\tilde{C}},s',h'\setminus h'_G\uplus h'_R,\rho',\Gamma, Q\ast R, A\cup X)$.

Suppose that the transition is given by $(E)$. Let $$A'=A\cup X \cup \bigcup_{\hat{r}\in Locked(\resk{r}{C})} PV(\hat{r}).$$

We have that $C'= \resk{r}{C}$, $(s,h\uplus h_R,\rho)\stackrel{A'}{\leftrightsquigarrow} (s',h\uplus h_R,\rho')$ and there exists $h'_G\subseteq h'$ such that $h'=h\uplus h_R\uplus h'_G$ and $$s',h'_G\models \underset{\hat{r}\in D'}{\circledast} \Gamma(\hat{r}).$$

From $Locked(\resk{r}{C})\cup \{r\}=Locked(C)$ and $PV(r)=X$, $$A'= A\cup \bigcup_{\hat{r}\in Locked(C)} PV(\hat{r}).$$

We have the environment transformation
$$(s,h,(O\cup\{r\},L,D))\stackrel{A'}{\leftrightsquigarrow} (s',h,(O'\cup\{r\},L',D')).$$
And the environment transition $$ C,(s,h\uplus h_G,(O\cup\{r\},L,D)) \xrightarrow{A, \Gamma'}_e C,(s',h\uplus h'_G,(O'\cup\{r\},L',D')).$$

From $Safe_{n+1}(C,s,h,(O\cup \{r\}, L,D),\Gamma', Q, A)$ and the environment transition above, it follows that $Safe_{n}(C,s',h,(O'\cup \{r\}, L',D'),\Gamma', Q, A)$ is valid.

By $P(n)$, $Safe_{n}(\resk{r}{C},s',h,\rho',\Gamma, Q\ast R, A\cup X)$ is valid.
\end{proof}

The soundness of the renaming rule follows from the next proposition which is a consequence of Proposition \ref{abortrename}.

\begin{prop}\label{saferename}
Let $C$ be a reachable command, $(s,h,(O,L,D))$ be a state, $A\subseteq \textbf{Var}$, $\Gamma$ a well-formed resource context and $r,r'$ resource names such that $r'\notin Res(C)$, $r'\notin Res(\Gamma)$ and $O\cup L\cup D=Res(\Gamma)$.

If $Safe_n(C[r'/r], s,h,\rho[r'/r], \Gamma[r'/r],Q,A)$ is valid, then $$Safe_n(C, s,h,\rho, \Gamma,Q,A)\textrm{ is valid.}$$
\end{prop}

\begin{proof}
We prove the proposition by induction on $n$. 
Let $C$ be a reachable command, $(s,h,(O,L,D)$ be a state, $A\subseteq \textbf{Var}$, $\Gamma$ a well-formed resource context and $r,r'\in\textbf{Res}$ such that $r'\notin Res(C)$, $r'\notin Res(\Gamma)$, $O\cup L\cup D=Res(\Gamma)$ and $Safe_n(C[r'/r], s,h,\rho[r'/r], \Gamma[r'/r],Q,A)$ is valid.

First, note that $r\notin \rho$, because $r'\notin Res(\Gamma)$ and $O\cup L\cup D=Res(\Gamma)$.

For $n=0$, it is trivially true.
Let $n=k+1$.

Because $Safe_n(C[r'/r], s,h,\rho[r'/r], \Gamma[r'/r],Q,A)$ is valid. If $C=\skipk$, then $C[r'/r]=\skipk$ and $s,h\models Q$, .

So the property $(i)$ of $Safe_n(C, s,h,\rho, \Gamma,Q,A)$ is verified.

From $Safe_n(C[r'/r], s,h,\rho[r'/r], \Gamma[r'/r],Q,A)$ and Proposition $\ref{abortrename}$, $$C[r'/r],(s,h,\rho[r'/r])\not\rightarrow_{p} \abok.$$ 

Hence, we have the property $(ii)$ of $Safe_n(C, s,h,\rho, \Gamma,Q,A)$.

From $Safe_n(C[r'/r], s,h,\rho[r'/r], \Gamma[r'/r],Q,A)$ be valid, we have $$chng(C)\cap \bigcup_{\hat{r}\in L\cup D} PV(\hat{r})=chng(C[r'/r])\cap \bigcup_{\hat{r}\in (L\cup D)[r'/r]} PV(\hat{r})=\emptyset.$$

Then the property $(iii)$ of $Safe_n(C, s,h,\rho, \Gamma,Q,A)$ is respected.

Let $h_G$, $C'$, $s'$, $h'$ and $\rho'$ such that $h_G\bot h$, $s,h_G\models \underset{\hat{r}\in D}{\circledast} \Gamma(\hat{r})$ and  
$$C,(s,h\uplus h_G,\rho)\xrightarrow{A, \Gamma} C', (s', h',\rho').$$

Note that $C'$ is a reachable command. By Proposition $\ref{abortrename}$, we have that $$C[r'/r],(s,h\uplus h_G,\rho[r'/r])\xrightarrow{A, \Gamma[r'/r]} C'[r'/r], (s', h',\rho'[r'/r]).$$

Moreover, we know that $$s,h_G\models \underset{\hat{r}\in D[r'/r]}{\circledast} \Gamma[r'/r](\hat{r}).$$ 

By $Safe_n(C[r'/r], s,h,\rho[r'/r], \Gamma[r'/r],Q,A)$ and the transition above, there exists $h'_G$ such that $h'_G\subseteq h'$,  $Safe_k(C'[r'/r], s',h'\setminus h'_G,\rho'[r'/r], \Gamma[r'/r],Q,A)$ is valid and $$s',h'_G\models \underset{\hat{r}\in D'[r'/r]}{\circledast} \Gamma[r'/r](\hat{r}).$$

It is easy to see that $s',h'_G\models \underset{\hat{r}\in D'}{\circledast} \Gamma(\hat{r})$.

By induction hypothesis, $Safe_k(C', s',h'\setminus h'_G,\rho', \Gamma,Q,A)$ is valid.
\end{proof}

To prove the soundness of the rule for auxiliary variables, we have the following result.

\begin{prop}\label{safeaux}
Let $C$ be a reachable command, $(s,h),(s',h)\in\mathcal{S}$, $\rho\in\mathcal{O}$, $Q$ an assertion, $A,X\subseteq\textbf{Var}$, $\Gamma$ a resource context and $l\in\mathbb{N}_0$ such that $X$ is a set of auxiliary variables for $C$, $l$ is the number of assignments to auxiliary variables in $C$, $FV(Q)\subseteq A$ and $X\cap (PV(\Gamma)\cup FV(Q))=\emptyset$.

If $Safe_{n+l}(C,s,h,\rho, \Gamma, Q, A\cup X)$ is valid and $s(x)=s'(x)$, for every $x\notin X$, then $Safe_{n}(C\setminus X,s',h,\rho, \Gamma, Q, A)$ is valid.
\end{prop}

\begin{proof}

Let $C$ be a reachable command, $(s,h),(s',h)\in\mathcal{S}$, $\rho\in\mathcal{O}$, $Q$ an assertion, $A,X\subseteq\textbf{Var}$, $\Gamma$ a resource context and $l\in\mathbb{N}_0$ such that $X$ is a set of auxiliary variables for $C$, $l$ is the number of assignments to auxiliary variables in $C$, $FV(Q)\subseteq A$, $X\cap (PV(\Gamma)\cup FV(Q))=\emptyset$, $s'(x)=s(x)$, for every $x\notin X$ and $Safe_{n+l}(C,s,h,\rho, \Gamma, Q, A\cup X)$ is valid.

The proof is done by induction on $n$. If $n=0$, the conclusion is valid.

Let $n=k+1$.

Suppose that $C\setminus X = \skipk$. Then $C=\skipk$ or $C=\assigk{x}{e}$, where $\textsf{x}\in X$. 

First, we assume $C=\skipk$. From $Safe_{n+l}(\skipk,s,h,\rho, \Gamma, Q, A\cup X)$, we have $s,h\models Q$.
If $s(y)=s'(y)$, for every $y\in FV(Q)\subset \textbf{Var}\setminus X$, then $s',h\models Q$.

Now, we assume $C=\assigk{x}{e}$, $\textsf{x}\in X$.
Consider the transition given by $(BCT)$ $$C,(s,h,\rho)\rightarrow_p \skipk, (s'',h,\rho),$$ where $s(y)=s''(y)$, for every $y\notin X$. 

Then $Safe_{n+l-1}(\skipk,s'',h,\rho, \Gamma, Q, A\cup X)$ is valid and $n+l-1>0$. Hence $s'',h\models Q$ and  $s',h\models Q$.

The above leads to the first property of $Safe_{n}(C\setminus X,s',h,\rho, \Gamma, Q, A)$.

From $Safe_{n+l}(C,s,h,\rho, \Gamma, Q, A\cup X)$, we know that the execution of $C$ does not abort, for $(s,h,\rho)$. The command $C\setminus X$ substitutes assignments by $\skipk$, so the execution of $C\setminus X$ does not abort, for $(s,h,\rho)$.
 Changing the value of auxiliary variables $X$ does not introduce aborts in the execution of $C\setminus X$. Hence $$C\setminus X, (s',h,\rho)\not\rightarrow_{p} \abok.$$
The second property of $Safe_{n}(C\setminus X,s',h,\rho, \Gamma, Q, A)$ is verified.

Note that $chng(C\setminus X)=chng(C)\setminus X$. From $Safe_{n+l}(C,s,h,\rho, \Gamma, Q, A\cup X)$, $$chng(C\setminus X)\cap \bigcup_{r\in L\cup D} PV(r)\subseteq chng(C)\cap \bigcup_{r\in L\cup D} PV(r)=\emptyset.$$
This establishes the third property of $Safe_{n}(C\setminus X,s',h,\rho, \Gamma, Q, A)$.

Let $h_G$, $\hat{C}$, $\hat{s}$, $\hat{h}$ and $\hat{\rho}$ such that $h_G\bot h$, $s',h_G\models \underset{r\in D}{\circledast} \Gamma(r)$ and  
$$C\setminus X,(s',h\uplus h_G,\rho)\xrightarrow{A, \Gamma} \hat{C}, (\hat{s}, \hat{h},\hat{\rho}).$$

Suppose that it is a program transition. 

We observe that there exists $\hat{C}'$, $\hat{s}'$ and $j$ such that $\hat{C}'\setminus X=\hat{C}$, $\hat{s}'(x)=\hat{s}(x)$, for every $x\notin X$, $j=l(C)-l(\hat{C}')$ and
$$C,(s,h\uplus h_G,\rho)\rightarrow_{p}^{j+1} \hat{C}', (\hat{s}', \hat{h},\hat{\rho}).$$


From $Safe_{n+ l}(C,s,h,\rho, \Gamma, Q, A\cup X)$, we know that there is $h'_G\subseteq \hat{h}$ such that $Safe_{n+l -j-1}(\hat{C}',\hat{s}', \hat{h}\setminus h'_G,\hat{\rho},\Gamma, Q, A\cup X)$ is valid and $\hat{s}',h'_G\models \underset{r\in \hat{D}}{\circledast} \Gamma(r)$.
Then $\hat{s},h'_G\models \underset{r\in \hat{D}}{\circledast} \Gamma(r)$, because $\hat{s}'(x)=\hat{s}(x)$, for every $x\in PV(\Gamma)\subseteq\textbf{Var}\setminus X$.

Note that $j$ is the number of assignments to auxiliary variables in $\hat{C}'$.

By the induction hypothesis, we conclude $Safe_{k}(\hat{C}'\setminus X, \hat{s}, \hat{h}\setminus h'_G,\hat{\rho}, \Gamma, Q, A)$.

Last, suppose that it is an environment transition.
Let $$A'=A\cup \bigcup_{r\in Locked(C\setminus X)} PV(r).$$

We have $\hat{C}=C\setminus X$,  $(s',h,\rho)\stackrel{A'}{\leftrightsquigarrow} (\hat{s}, h,\hat{\rho})$, and there exists $h'_G\subseteq \hat{h}$ such that $\hat{h}=h\uplus h'_G$ and $$\hat{s},h'_G\models \underset{r\in \hat{D}}{\circledast} \Gamma(r).$$

Note that $Locked(C)=Locked(C\setminus X)$. And consider the storage $\hat{s}'$ such that $\hat{s}'(x)=\hat{s}(x)$, if $x\notin X$, and $\hat{s}'(x)=s(x)$, if $x\in X$ and $$A''=A\cup X\cup \bigcup_{r\in Locked(C)} PV(r).$$

We have the following environment transformation $(s,h,\rho)\stackrel{A''}{\leftrightsquigarrow} (\hat{s}', h,\hat{\rho})$,
and the environment transition
$$C, (s,h,\rho) \xrightarrow{A\cup X, \Gamma}_e C,(\hat{s}', h,\hat{\rho}).$$

By the previous transition, we obtain $Safe_{k+l}(C,\hat{s}',h,\hat{\rho}, \Gamma, Q, A\cup X)$.

Note that $\hat{s}'(x)=\hat{s}(x)$, for every $x\notin X$.  Therefore by the induction hypothesis, we have that  $Safe_{k}(C\setminus X, \hat{s}, h,\hat{\rho}, \Gamma, Q, A)$ is valid.
\end{proof}

The last proposition exhausts the inferences rules of CSL, completing the proof of Theorem~\ref{CSLtosafe}. 

\section{Conclusion}
This work presents a proof of correctness of CSL based on SOS, the first we are
aware of. We build on two previous proofs, one for
the full logic, using a denotational semantics based on traces, and
another for a fragment of CSL, the DCSL. An immediate extension to this work is its formalization in a theorem prover.


A proof based on SOS is important, as this form of
semantics closer mimics the execution of an imperative
program. Therefore, it paves the way to the development of more
expressive proving tools that are able to deal with truly concurrent
programs manipulating shared resources. 
Our work may also provide insight on how to develop provably correct compilers able of detecting data-races.

Our aim was lifting the (severe) restriction of forcing concurrent
threads to manipulate only disjoint sets of variables, since it does
not allow proving correct many interesting and useful programs. To
attain this goal, we re-used the notion of ``rely-set'' 
, a notion crucial
to obtain the soundness result of CSL with respect to the denotational
semantics. The adaptation was not trivial and required developing
several auxiliary notions, but established a proof technique that may
now be used in other contexts.



\section*{Appendix}\label{Appendix}

\begin{proof}[Proof of Proposition~\ref{safetymono}]
Let $C\in \mathcal{C}$, $(s,h)\in \mathcal{S}$, and $\rho\in\mathcal{O}$. Suppose $h_F$ is a heap such that $h\bot h_F$ and $$C,(s,h\uplus h_F,\rho)\rightarrow_{p} \abok.$$

We'll prove that $C,(s,h,\rho)\rightarrow_{p} \abok$ by induction on the relation, $\rightarrow_{p}$.

Suppose that the transition to $\abok$ is given by $(RA)$, $(WA)$ or $(WA2)$.  
Then the transitions does not depended on the heap and  $C,(s,h,\rho)\rightarrow_{p} \abok
$.

Suppose that it is given by $(BCA)$. Then there is a faulty memory access of $h\uplus h_F$ and, consequently, of $h$. Hence $C,(s,h,\rho)\rightarrow_{p} \abok$.

Suppose that it is given by $(SA)$.
Then $C=\seqk{C_1}{C_2}$ and $C_1, (s,h\uplus h_F,\rho)\rightarrow_p \abok$.

By induction, we know that $C_1, (s,h,\rho)\rightarrow_p \abok$. Hence $C, (s,h,\rho)\rightarrow_p \abok$.

The remaining cases are similar to the previous case.
\end{proof}

\begin{proof}[Proof of Proposition~\ref{frameprop}]

Let  $C,C'\in \mathcal{C}$, $(s,h), (s',h')\in\mathcal{S}$, $\rho,\rho'\in\mathcal{O}$ and $h_F$ such that $h\bot h_F$, $C,(s,h,\rho')\not\rightarrow_{p} \abok$ and $$C,(s,h\uplus h_F,\rho)\rightarrow_{p} C',(s',h',\rho').$$

We prove by induction on the program rules that $h_F$ is a subheap of $h'$ and $C,(s,h,\rho)\rightarrow_{p} C',(s',h'\setminus h_F,\rho')$.

Suppose that the transition is given by $(BCT)$. Because the Separation Logic respects the frame property, we know that $h_F$ is a subheap of $h'$ and $$C,(s,h,\rho)\rightarrow_{p} C',(s',h'\setminus h_F,\rho').$$

Suppose that the transition is given by $(S1)$, $(LP)$, $(IF1)$, $(IF2)$, $(R0)$, $(P3)$, $(W0)$ or $(W2)$. Then the transition neither depends nor changes the heap function. So we obtain that $h'$ and $$C,(s,h,\rho)\rightarrow_{p} C',(s',h'\setminus h_F,\rho').$$

Suppose that the transition is given by $(S2)$. Then $C=\seqk{C_{1}}{C_{2}}$ and $C'=\seqk{C'_{1}}{C_{2}}$ such that $$C_1,(s,h\uplus h_F,\rho)\rightarrow_{p} C'_1,(s',h',\rho').$$ 

If $C_1,(s,h,\rho)\rightarrow \abok$, then $C,(s,h,\rho)\rightarrow \abok$. Hence $$C_1,(s,h,\rho)\not\rightarrow \abok.$$

From the induction hypothesis, we conclude that $h_F\subseteq h'$ and $$C_1,(s,h,\rho)\rightarrow_{p} C'_1,(s',h'\setminus h_F,\rho').$$

Therefore $$C,(s,h,\rho)\rightarrow_{p} C',(s',h'\setminus h_F,\rho').$$

The cases $(P1)$, $(P2)$, $(R1)$, $(R2)$ and $(W1)$ are similar to the previous case.
\end{proof}

\begin{thebibliography}{10}

\bibitem{ben}
M.~Ben-Ari.
\newblock {\em Principles of concurrent and distributed programming(Second
  Edition)}.
\newblock Addison-Wesley, 2006.

\bibitem{boy}
J.~Boyland.
\newblock Checking interference with fractional permissions.
\newblock In {\em SAS}, volume 2694 of {\em Lecture Notes in Computer Science},
  pages 55--72. Springer, 2003.

\bibitem{bro1}
S.~Brookes.
\newblock A semantics for concurrent separation logic.
\newblock {\em Theoretical Computer Science}, 375(1-3):227--270, 2007.

\bibitem{bro2}
S.~Brookes.
\newblock A revisionist history of concurrent separation logic.
\newblock {\em ENTCS}, 276:5--28, 2011.

\bibitem{CDDGHLOP15}
C.~Calcagno, D.~Distefano, J.~Dubreil, D.~Gabi, P.~Hooimeijer, M.~Luca, P.~W.
  O'Hearn, I.~Papakonstantinou, J.~Purbrick, and D.~Rodriguez.
\newblock Moving fast with software verification.
\newblock In {\em {NASA} Formal Methods - 7th International Symposium, {NFM}
  2015, Pasadena, CA, USA, April 27-29, 2015, Proceedings}, pages 3--11, 2015.

\bibitem{hoar}
C.~A.~R. Hoare.
\newblock An axiomatic basis for computer programming.
\newblock {\em Communications of the ACM}, 12(10):576--580, 1969.

\bibitem{hearn}
P.~W. O'Hearn.
\newblock Resources, concurrency, and local reasoning.
\newblock {\em Theoretical Computer Science}, 375(1-3):271--307, 2007.

\bibitem{rey2}
P.~W. O'Hearn, J.~C. Reynolds, and H.~Yang.
\newblock Local reasoning about programs that alter data structures.
\newblock In {\em CSL}, volume 2142 of {\em Lecture Notes in Computer Science},
  pages 1--19. Springer, 2001.

\bibitem{owi}
S.~S. Owicki.
\newblock A consistent and complete deductive system for the verification of
  parallel programs.
\newblock In {\em STOC}, pages 73--86. ACM, 1976.

\bibitem{owigrie}
S.~S. Owicki and D.~Gries.
\newblock Verifying properties of parallel programs: An axiomatic approach.
\newblock {\em Communications of the ACM}, 19(5):279--285, 1976.

\bibitem{plot}
G.~D. Plotkin.
\newblock A structural approach to operational semantics.
\newblock {\em Journal of Logic and Algebraic Programming}, 60--61:17--139,
  2004.

\bibitem{red}
U.~S. Reddy and J.~C. Reynolds.
\newblock Syntactic control of interference for separation logic.
\newblock In {\em POPL}, pages 323--336. ACM, 2012.

\bibitem{rey}
J.~C. Reynolds.
\newblock Separation logic: A logic for shared mutable data structures.
\newblock In {\em LICS}, pages 55--74. IEEE Computer Society, 2002.

\bibitem{RR-DCC-11-2014}
P.~Soares, A.~Ravara, and S.~Melo de~Sousa.
\newblock An operational semantics for concurrent separation logic.
\newblock Technical Report RR-DCC-2014-11, Department of Computer Science,
  Faculty of Science, University of Porto, 2014.

\bibitem{pdp-4pad15}
P.~Soares, A.~Ravara, and S.~Melo de~Sousa.
\newblock Revisiting concurrent separation logic and operational semantics.
\newblock In {\em PDP}, pages 484--491. IEEE, 2015.

\bibitem{vaf}
V.~Vafeiadis.
\newblock Concurrent separation logic and operational semantics.
\newblock {\em ENTCS}, 276:335--351, 2011.

\bibitem{vaf2}
V.~Vafeiadis and M.~J. Parkinson.
\newblock A marriage of rely/guarantee and separation logic.
\newblock In {\em CONCUR}, volume 4703 of {\em Lecture Notes in Computer
  Science}, pages 256--271. Springer, 2007.

\end{thebibliography}
\end{document}